\documentclass{arxiv}
\usepackage{graphicx,amssymb,amsmath}
\usepackage{todonotes}
\usepackage{subcaption}
\usepackage{soul}
\usepackage{xspace}
\usepackage{paralist}
\usepackage{hyperref}
\usepackage{ulem}
\usepackage{xcolor}
\usepackage{thmtools} 
\usepackage{thm-restate}
\usepackage{hyperref}

\declaretheorem[name=Theorem]{thm}
\declaretheorem[name=Lemma]{lem}

\title{Shortcut Hulls: Vertex-restricted Outer Simplifications of Polygons}

\author{Annika Bonerath\thanks{Institute of Geodesy and Geoinformation, University of Bonn, {\tt bonerath@igg.uni-bonn.de}}
        \and
        Jan-Henrik Haunert\thanks{Institute of Geodesy and Geoinformation, University of Bonn, {\tt haunert@igg.uni-bonn.de}}
        \and
        Joseph S. B. Mitchell\thanks{Stony Brook University {\tt joseph.mitchell@stonybrook.edu}}
        \and Benjamin Niedermann\thanks{Institute of Geodesy and Geoinformation, University of Bonn, {\tt niedermann@igg.uni-bonn.de}}}

\begin{document}
\thispagestyle{empty}
\maketitle

\newcommand{\area}{c_\mathrm{A}}
\newcommand{\peri}{c_\mathrm{P}}
\newcommand{\cost}{c}
\newcommand{\length}{c_\mathrm{P}}

\newcommand{\probBendHull}{\textsc{$k$-Bend\-Short\-cut\-Hull}\xspace}
\newcommand{\probEdgeHull}{\textsc{$k$-Edge\-Short\-cut\-Hull}\xspace}

\def\ProblemShortcutHull{\par\noindent{\textsc{ShortcutHull}.\ } \ignorespaces}
\def\endProblemShortcutHull{}
\def\ProblemBendShortcutHull{\par\noindent{\probBendHull.\ } \ignorespaces}
\def\endProblemBendShortcutHull{}
\def\ProblemEdgeShortcutHull{\par\noindent{\probEdgeHull.\ } \ignorespaces}
\def\endProblemEdgeShortcutHull{}

\newcommand{\tighthulledges}{enrichment\xspace}
\newcommand{\tighthulledgesPlural}{enrichments\xspace}
\newcommand{\extendedEdgeSet}{enrichment\xspace}
\newcommand{\containingPoly}{containing box\xspace}
\newcommand{\donutPoly}{sliced donut\xspace}
\newcommand{\slicedTightHull}{sliced tight hull\xspace}
\newcommand{\rem}[1]{\textcolor{red}{\sout{#1}}}
\newcommand{\add}[1]{\textcolor{blue}{\uline{#1}}}
\newcommand{\move}[1]{\textcolor{pink}{\uwave{#1}}}

\normalem

\begin{abstract}
Let $P$ be a crossing-free polygon and $\mathcal C$ a set of shortcuts, where each shortcut is a directed straight-line segment connecting two vertices of $P$.
A shortcut hull of $P$ is another crossing-free polygon that encloses $P$ and whose oriented boundary is composed of elements from $\mathcal C$.
Shortcut hulls find their application in geo-related problems such as the simplification of contour lines. 
We aim at a shortcut hull that linearly balances the enclosed area and perimeter.
If no holes in the shortcut hull are allowed,
the problem admits a straight-forward solution via shortest paths.
For the more challenging case that the shortcut hull may contain holes, we present a polynomial-time algorithm that is based on computing a constrained, weighted triangulation of the input polygon's exterior. 
We use this problem as a starting point for investigating further variants, e.g., restricting the number of edges or bends. We demonstrate that shortcut hulls can be used for drawing the rough extent of point sets as well as for the schematization of polygons.
\end{abstract}

\section{Introduction}
The simplification of polygons finds a great number of applications in geo-related problems. 
For example in map generalization it is used to obtain abstract representations of area features such as lakes, buildings, 
or contour lines. 
A common technique, which originally  stems from polyline simplification, is to restrict the resulting polygon~$Q$ of a polygon $P$ to the vertices of $P$, which is also called a \emph{vertex-restricted simplification} \cite{driemel2013jaywalking,filtser2021static,meulemans2014}.
In that case $Q$ consists of straight edges\footnote{Throughout this paper, we use the term \emph{edge} instead of \emph{straight-line segment}.} that are \emph{shortcuts} between vertices of $P$. 
In the classic problem definition of line and area simplification the result $Q$ may cross edges of $P$. 

\begin{figure}[t]
    \begin{subfigure}[c]{0.32\linewidth}
        \includegraphics[page=1,width=\textwidth]{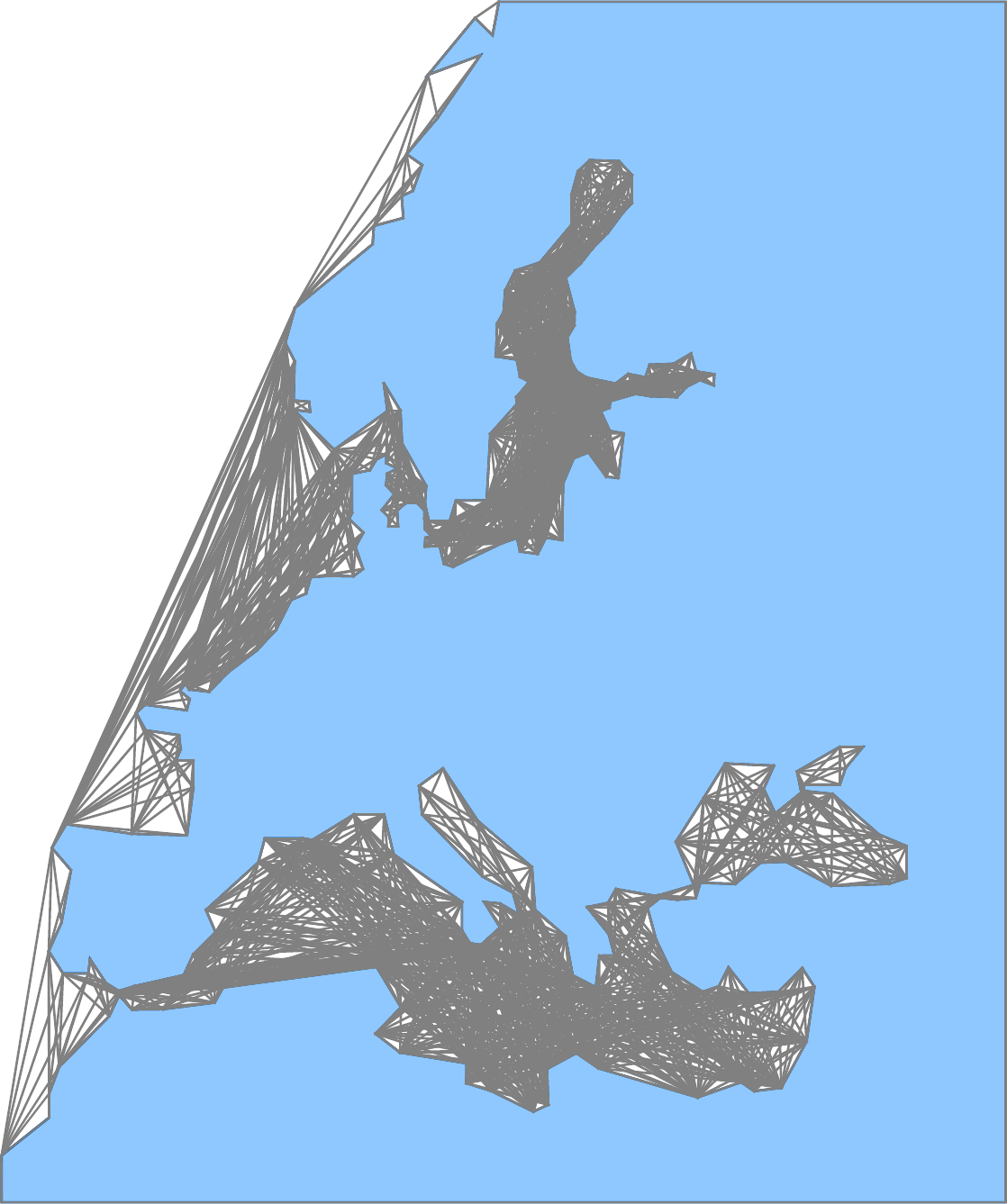}
     \caption{}
    \end{subfigure}
    \begin{subfigure}[c]{0.32\linewidth}
     \includegraphics[page=1,width=\textwidth]{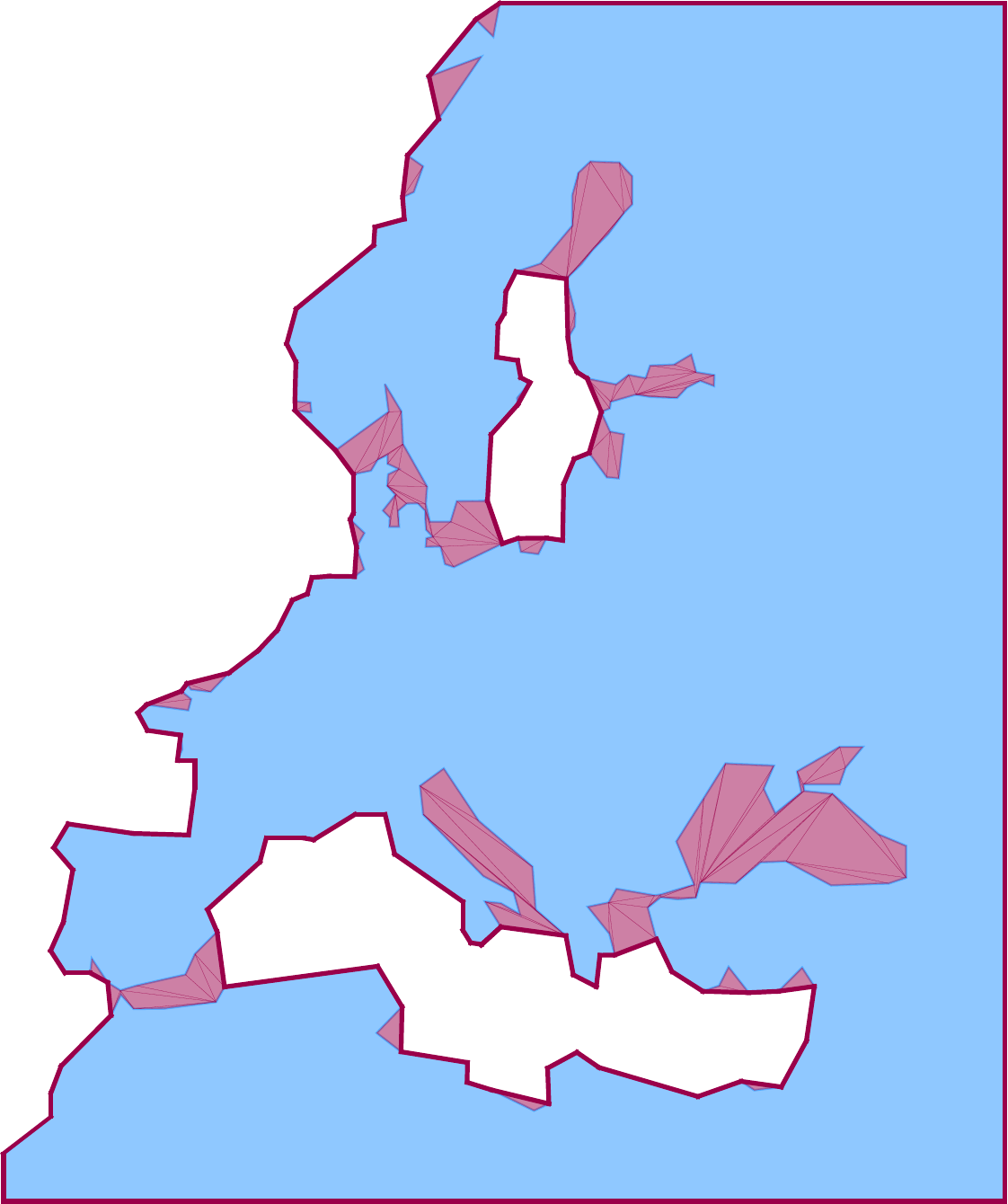}
     \caption{}
    \end{subfigure}
    \begin{subfigure}[c]{0.32\linewidth}
     \includegraphics[page=1,width=\textwidth]{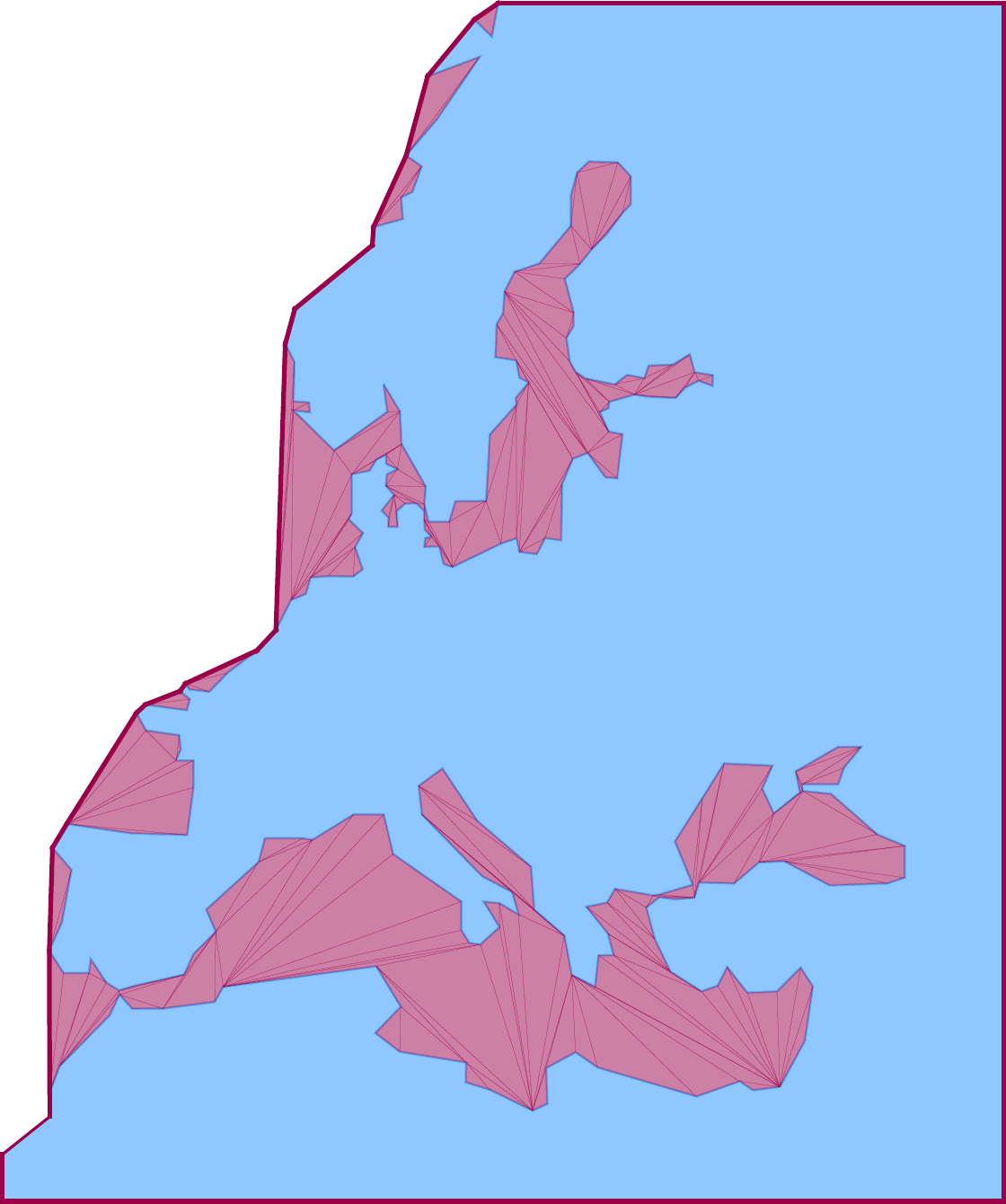}
     \caption{}
    \end{subfigure}
    
    \smallskip
    \begin{subfigure}[c]{0.32\linewidth}
        \includegraphics[page=1,width=\textwidth]{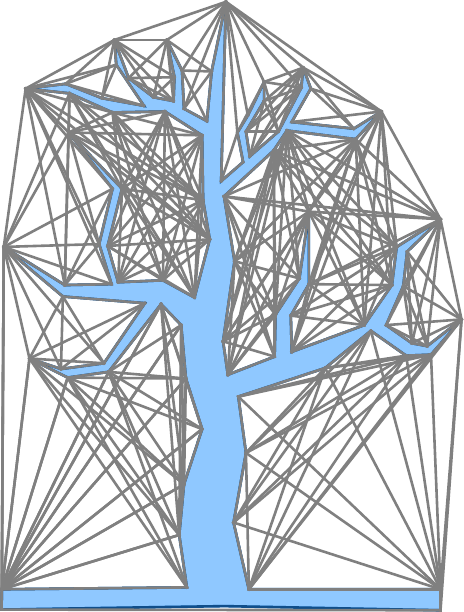}
     \caption{}
    \end{subfigure}
    \begin{subfigure}[c]{0.32\linewidth}
     \includegraphics[page=1,width=\textwidth]{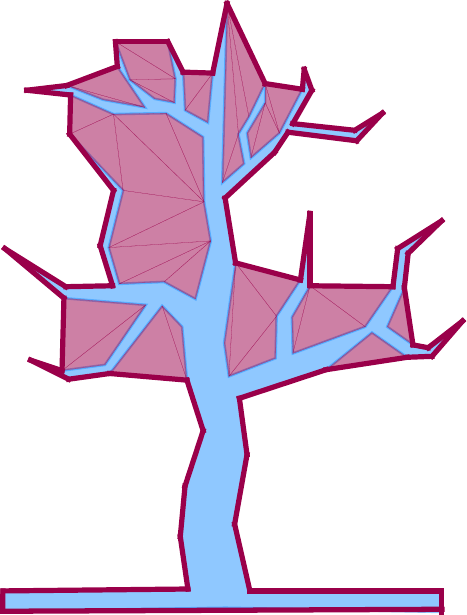}
     \caption{}
    \end{subfigure}
    \begin{subfigure}[c]{0.32\linewidth}
     \includegraphics[page=1,width=\textwidth]{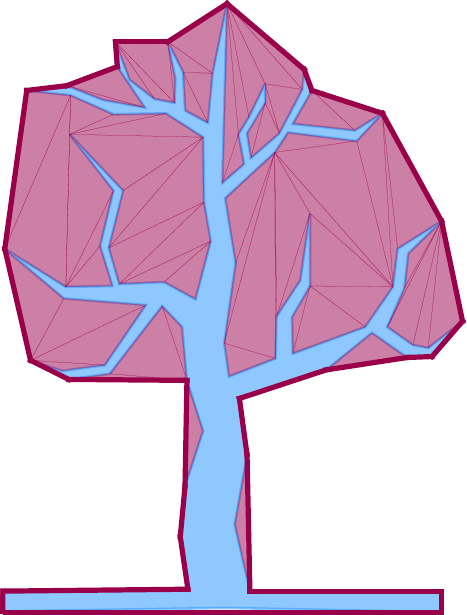}
     \caption{}
    \end{subfigure}
    
    \caption{1st column: Input polygon (blue) with a set $\mathcal C$ of all possible shortcuts (gray). 2nd--3rd columns: Optimal $\mathcal C$-hulls (blue and red area) for different $\lambda$.}
    \label{fig:example:output}
\end{figure}
In this paper, we consider the vertex-restricted crossing-free simplification of a polygon~$P$ considering only shortcuts that lie in the exterior of $P$ or are part of the boundary of $P$. 
In contrast to other work, we consider the shortcuts as input for our problem and do not require special properties, e.g., that they are crossing-free, or that they comprise all possible shortcuts. 
The result of the simplification is a \emph{shortcut hull}~$Q$ of $P$ possibly having holes. 
We emphasize that the edges of a shortcut hull do not cross each other.  
\autoref{fig:example:output} shows polygons (blue area) with all possible shortcuts and different choices of shortcut hulls (blue and red area).
Such hulls find their application when it is important that the simplification contains the polygon. Figure~\ref{fig:example:lakes} shows the simplification of a network of lakes. We emphasize that the lakes are connected to the exterior of the green polygon at the bottom side. In that use case, it can be desirable that the water area is only decreased to sustain the area of the land occupied by important map features. The degree of the simplification of $Q$ can be measured by its perimeter and enclosed area. While a small perimeter indicates a strong simplification of $P$, a small area gives evidence that $Q$ adheres to $P$. In the extreme case $Q$ is either the convex hull of $P$ minimizing the possible perimeter, or $Q$ coincides with $P$ minimizing the enclosed area. We present algorithms that construct shortcut hulls of $P$ that linearly balance these two contrary criteria by a parameter $\lambda\in [0,1]$, which specifies the degree of simplification. With increasing $\lambda$ the enclosed area is increased, while the perimeter is decreased.
We show that for the case that $Q$ must not have holes we can reduce the problem to finding a cost-minimal path in a directed acyclic graph that is based on the given set of possible shortcuts. However, especially for the application in geovisualization, where it is about the simplification of spatial structures, we deem the support of holes in the simplification as an essential key feature. For example, in Figure~\ref{fig:example:lakes:output} the connections between the lakes are not displayed anymore as they are very narrow, while it is desirable to still show the large lakes. We therefore investigate the case of shortcut hulls with holes in greater detail. 

\begin{figure}[t]
    \centering
    \begin{subfigure}[c]{0.45\linewidth}
        \includegraphics[page=1,width=\textwidth]{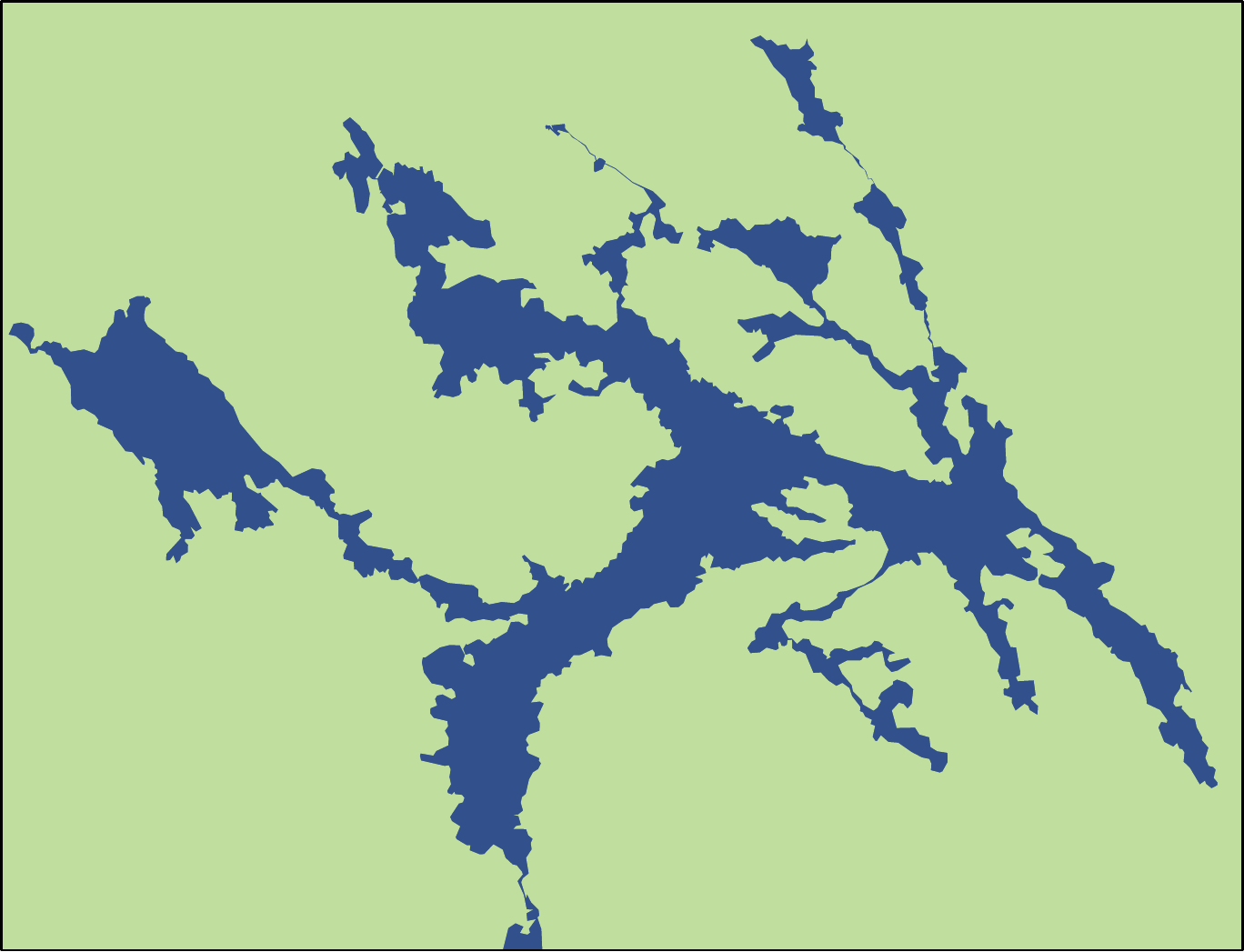}
     \caption{input map}
    \end{subfigure}
    \begin{subfigure}[c]{0.45\linewidth}
      \includegraphics[page=2,width=\textwidth]{figs/example-lakes-procedure.pdf}
     \caption{input polygon~$P$}
     \label{fig:example:lakes:without-holes}
    \end{subfigure}
    
\begin{subfigure}[c]{0.45\linewidth}
      \includegraphics[page=3,width=\textwidth]{figs/example-lakes-procedure.pdf}
     \caption{optimal shortcut hull~$Q$}
     \label{fig:example:lakes:holes}
    \end{subfigure}
    \begin{subfigure}[c]{0.45\linewidth}
      \includegraphics[page=4,width=\textwidth]{figs/example-lakes-procedure.pdf}
     \caption{simplified map}
     \label{fig:example:lakes:output}
    \end{subfigure}

    \caption{Simplification of a network of lakes in Sweden.
    }
    \label{fig:example:lakes}
\end{figure}

\begin{figure}[t]
    \centering
    \begin{subfigure}[c]{0.31\linewidth}
     \includegraphics[page=2,width=\textwidth]{figs/weaklysimple.pdf}
     \caption{}
    \end{subfigure}
    \begin{subfigure}[c]{0.31\linewidth}
     \includegraphics[page=1,width=\textwidth]{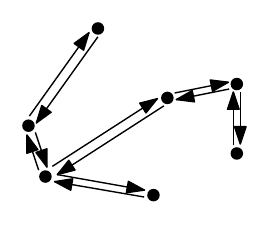}
     \caption{}
    \end{subfigure}
    \begin{subfigure}[c]{0.31\linewidth}
     \includegraphics[page=3,width=\textwidth]{figs/weaklysimple.pdf}
    \caption{}
    \end{subfigure}
    \caption{Weakly-simple polygons. (a)--(b) Valid input polygon as the exterior is a connected region. (c) Invalid input polygon as the exterior consists of two regions.}
    \label{fig:intro:weaklysimplepoly}
\end{figure}

\paragraph*{Input Polygon.} As input we expect a clockwise-oriented polygon~$P$ that is \emph{weakly-simple}, which means that we allow
vertices to lie in the interior of edges
as well as edges that point in opposite directions to lie on top of each other; see Figure~\ref{fig:intro:weaklysimplepoly}. In particular, the edges of $P$ do not cross each other. Such polygons are more general than simple polygons and can be  used to describe more complex geometric objects such as the faces of a graph embedded into the plane; see Figure~\ref{fig:application-examples} for minimum spanning tree. 
For the input polygon $P$ we further require that its exterior is one connected region; we say that the exterior of $P$ is \emph{connected}; see Figure~\ref{fig:intro:weaklysimplepoly}. 
Hence, both a simple polygon and the outer face of the plane embedding of a planar graph are possible inputs. Finally, we emphasize that $P$ may have holes. We can handle every 
 hole separately assuming that we have inserted a narrow channel in $P$ connecting it with the exterior of $P$; consider the lakes in  
 Figure~\ref{fig:example:lakes}.  We can force the algorithm to 
 fill the artificially introduced channel with the interior of $Q$.

\begin{figure}
    \centering
    \includegraphics[page=1,height=0.4\linewidth, angle =90]{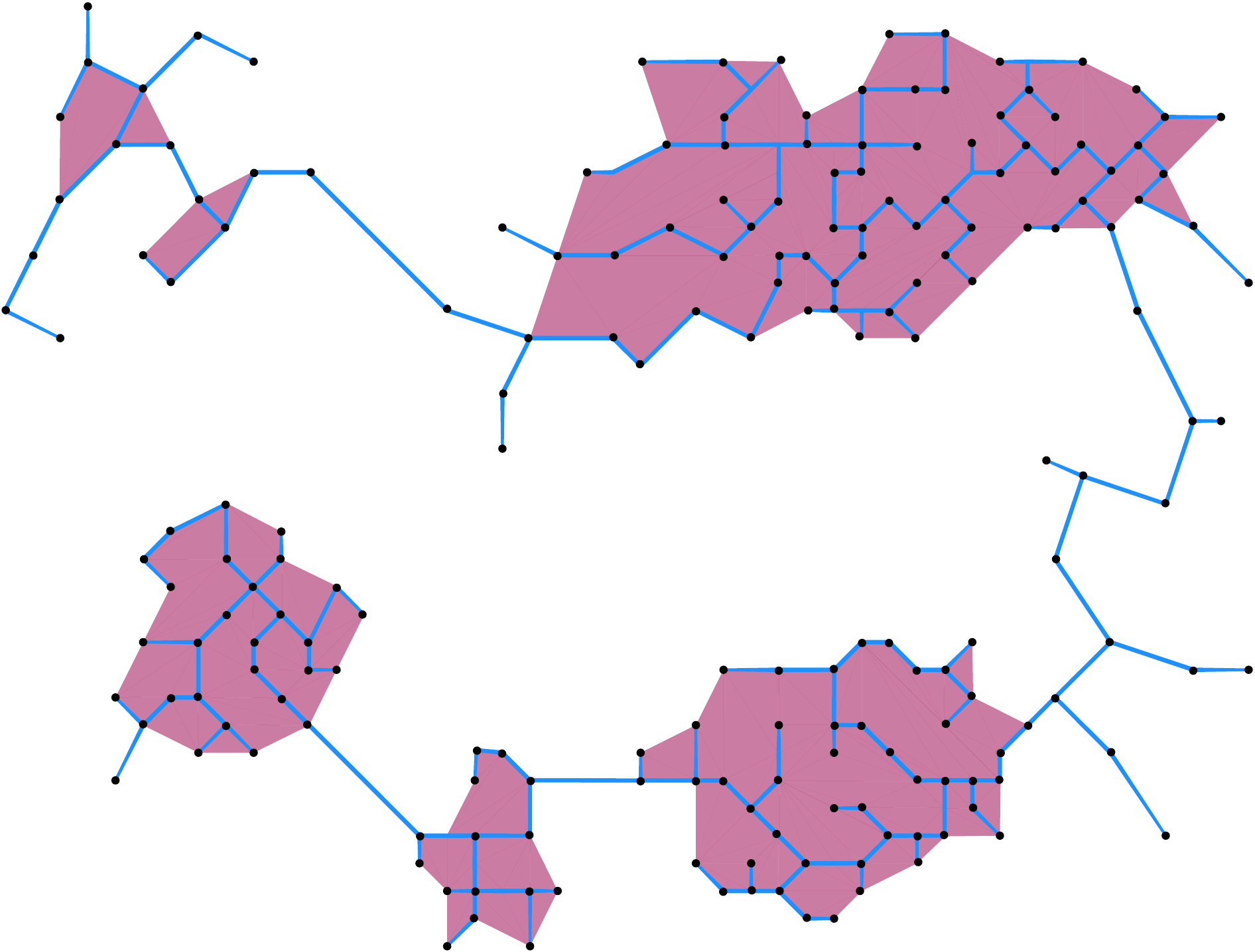}
    \caption{Shortcut hull of a minimum spanning tree.}
     \label{fig:example:mst}
    \label{fig:application-examples}
\end{figure}

\paragraph{Formal Problem Definition.}
We are given a weakly-simple polygon $P$ with connected exterior and a set $\mathcal C$ of directed edges in the exterior of $P$ such that the endpoints of the edges in $\mathcal C$ are vertices of $P$; see \autoref{fig:intro:inputA}. We call the elements in $\mathcal C$ \emph{shortcuts}.
A \emph{$\mathcal C$-hull} is a weakly-simple polygon whose oriented boundary  consists only of directed edges from $\mathcal C$, whose exterior is connected, and that contains $P$.
We allow $\mathcal C$-hulls to have holes. We observe that such holes can only lie in the exterior of $P$.
We are interested in a $\mathcal C$-hull~$Q$ that linearly balances the perimeter and enclosed area of $Q$. Formally, we define the \emph{cost} of a $\mathcal C$-hull~$Q$ as  
  \begin{equation}\label{equ:costs}
    \cost(Q) = \lambda \cdot \peri(Q)+ (1-\lambda) \cdot \area(Q),
  \end{equation}
  where $\lambda\in[0,1]$ is a given constant balancing the perimeter $\peri(Q)$ and the area $\area(Q)$ of $Q$. Further, $Q$ is \emph{optimal} if for every $\mathcal C$-hull~$Q'$ of $P$ it holds $\cost(Q)\leq \cost(Q')$.
\vspace{0.3em}
  \begin{ProblemShortcutHull}\label{problem1}\ \vspace{0.5ex} \\%
     \begin{tabular}{ll}
        \textbf{given:}  & weakly-simple polygon $P$ with $n$ vertices\\& and connected exterior, set $\mathcal C$ of shortcuts\\& of $P$, and $\lambda\in [0,1]$  \\
        \textbf{find:} &  optimal $\mathcal C$-hull $Q$ of $P$ (if it exists) 
     \end{tabular}\\
  \end{ProblemShortcutHull}
Further, we observe that it holds $|\mathcal C|\in O(n^2)$ as the edges of $\mathcal C$ have their endpoints on the boundary of $P$.

  \begin{figure}[t]
    \centering
    \begin{subfigure}[c]{0.32\linewidth}
     \includegraphics[page=1,width=\textwidth]{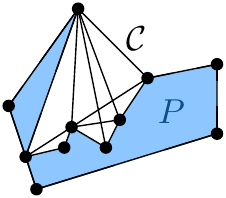}
     \caption{input $P$ and $\mathcal C$}
     \label{fig:intro:inputA}
    \end{subfigure}
    \begin{subfigure}[c]{0.32\linewidth}
     \includegraphics[page=2,width=\textwidth]{figs/probleminput.pdf}
     \caption{$\mathcal C$-hull $Q$}
     \label{fig:intro:inputB}
    \end{subfigure}
    \begin{subfigure}[c]{0.32\linewidth}
     \includegraphics[page=3,width=\textwidth]{figs/probleminput.pdf}
     \caption{pocket $P[e]+e$}
     \label{fig:intro:inputBC}
    \end{subfigure}
    \caption{The input, a solution, and a subinstance for an instance of the problem.
}
    \label{fig:intro:input}
\end{figure}
\begin{figure}
    \centering
    \begin{subfigure}[t]{0.32\linewidth}
    \includegraphics[page=3, width=\textwidth]{figs/spatial-complexity.pdf}
    \caption{$\chi=1$}
      \label{fig:triangulation-c-hull}
    \end{subfigure}
    \hspace{2em}
    \begin{subfigure}[t]{0.32\linewidth}
    \includegraphics[page=1, width=\textwidth]{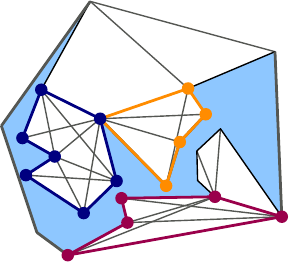}
    \caption{$\chi=7$}
    \label{fig:spatial-complexityA}
    \end{subfigure}
    \caption{
    Two examples of set $C$ with different spatial 
    complexities $\chi$. 
    (a) $\mathcal C$-triangulation and $\mathcal C$-hull. (b) connected components of the crossing graph.}
    \label{fig:spatialcomplexity}
\end{figure}
  
 \paragraph*{Our Contribution.}
 We first discuss how to construct an optimal $\mathcal C$-hull in $O(|\mathcal C|)$ time  for the case that it must not have holes (Section~\ref{sec:algorithm:without-holes}). Afterwards, we turn our focus to $\mathcal C$-hulls that may have holes (Sections~\ref{sec:structural-results}--\ref{sec:variants}).
 In particular, we show that finding an optimal $\mathcal C$-hull~$Q$ of~$P$ is closely related to finding a triangulation~$T$ of the exterior of~$P$ such that each triangle $\Delta \in T$ either belongs to the interior or exterior of~$Q$; see Figure~\ref{fig:triangulation-c-hull}. 
 We present an algorithm that solves \textsc{ShortcutHull}
 in $O(n^2)$ time if we forbid holes and in $O(n^3)$ time in the general case.  Moreover, in the case that the edges of $\mathcal C$ do not cross each other, it runs in $O(n)$ time.
 More generally, we analyse the running time based on the structure of $\mathcal C$. Let $S$ be the region between $P$ and the convex hull of $P$. Let $G$ be the \emph{crossing graph} of $\mathcal C$, i.e., each node of $G$ corresponds to an edge in $\mathcal C$ and two nodes of $G$ are adjacent if the corresponding edges in $\mathcal C$ cross each other. The \emph{spatial complexity} of $\mathcal C$ is the smallest number $\chi\in \mathbb N$ for which every connected component of $G$ can be enclosed by a polygon with $\chi$ edges that lies in the exterior of $P$ and only consists of vertices from $P$; see  Figure~\ref{fig:spatialcomplexity}.
 We show that the purposed algorithm runs in $O(\chi^3+n \chi)$ time. 
 We emphasize that $\chi\in O(n)$. Moreover, we present two variants of $\mathcal C$-hulls that restrict the 
 number of permitted edges or bends. 
 We further discuss relations of shortcut hulls with respect to problems 
 from application in cartography and computational geometry 
 (Section~\ref{sec:applications}).

\section{Related Work}
In the following, we consider two major research fields that are closely related to our work. 
At first, the field of representing geometric objects by less complex and possibly schematized geometric objects and, secondly, the field of constrained and weighted triangulations. 
Application fields for the representation of geometric objects by less-complex and possibly schematized objects are found, for example, in cartography: administrative borders~\cite{Barkowsky2000, Buchin2016, garcia1994boundary, Dijk2014}, building footprints~\cite{Haunert2010, van2013alpha}, and metro maps~\cite{Jacobsen2020, Noellenburg2014,  Wu2019}.
In particular, we want to point out the generalization of isobathymetric lines
in sea charts 
where the simplified line should lie on the downhill side of the original line to avoid the elimination
of shallows~\cite{zhang2011multi}.
In this context, it is important to find a good balance between the preservation of the information and the legibility of the visualization~\cite{burghardt2007investigations}.
Considering a polygon as input geometry, a basic technique for  simplification and schematization is the convex hull~\cite{alegria2020efficient,Daymude2020, Fink2004, rawlins1987optimal}.
An approach for rectilinear input polygons are tight rectilinear hulls~\cite{Bonerath2020TightHulls}.
Multiple other approaches for polygonal hulls of polygons exist---some of them can be solved in polynomial time~\cite{Haunert2010}, while others are shown to be NP-hard~\cite{haunert2008optimal}.
A closely related field is the
topologically correct
simplification and schematization of polygonal
subdivisions~\cite{Buchin2016,Estkowski2001,Mendel2018,Meulemans2010,	vanGoethem2015}. 
For the case that multiple geometric objects are the input of the problem, there exist several techniques for combining the aggregation and representation by a more simple geometry. 
In the case that the input is a set of polygons, a common technique is to use a partition of the plane, such as a triangulation, as basis~\cite{damen2008high,jones1995map,li2018automated,li2010triangulated,sayidov2019generalization,steiniger2006recognition}. 
In the case that the input is a set of points, we aim at representing this by a polygonal hull. 
Many approaches such as   
$\alpha$-shapes~\cite{DBLP:journals/tit/EdelsbrunnerKS83} and $\chi$-shapes~\cite{duckham2008efficient} use a triangulation as their basis.
Another approach is based on shortest-paths~\cite{DBLP:conf/gis/BergMS11}.
Note that there also exists work on combining the aggregation of point sets resulting in schematized polygons~\cite{Bonerath2019,van2013alpha}.
For considering polylines as input there exists work on computing an enclosing simple polygon based on the Delaunay triangulation~\cite{ai2017envelope}.
The schematization of
polylines is also closely related to our approach.
On the one hand, there is the schematization of a polyline
inside a polygon or between
obstacles~\cite{Adegeest1994, Lee1996,
Mitchell2014,Speckmann2018}. 
Alternatively, there also exists work on the simplification of a polyline based on a Delaunay triangulation ~\cite{ai2017envelope,ai2006hierarchical,ai2014simplification}. 
For the general 
simplification of polylines we also refer to 
the Douglas-Peucker algorithm, which is most widely applied in cartography~\cite{douglas1973algorithms}, and similar
approaches~\cite{Abam2010,Neyer1999,Pallero2013}.

Triangulating a polygon is widely studied in computational geometry. Triangulation of a simple polygon can be done in worst-case linear time~\cite{chazelle1991triangulating}. A polygon with $h$ holes, having in total $n$ vertices, can be triangulated in $O(n \log n)$ time~\cite{garey1978triangulating} or even $O(n+h \log^{1+\varepsilon} h)$ time~\cite{bar-yehuda-and-chazelle}. 
Our approach is particularly related to minimum-weight triangulations~\cite{shamos1975closest} and constrained triangulations~\cite{chew1989constrained,chin1998finding,kao1992incremental,lee1986generalized,shewchuk2015fast}. 

\section{Computing Optimal Shortcut Hulls without Holes} \label{sec:algorithm:without-holes}
Let $G_\mathcal{C}$ be the graph induced by the edges in $\mathcal C$. 
We call $G_\mathcal{C}$ the \emph{geometric graph} of $\mathcal C$.
If we do not allow the shortcut hull to have holes, we can 
compute an optimal $\mathcal C$-hull~$Q$ based on a cost-minimal path in $G_{\mathcal C}$; see Figure~\ref{fig:intro:inputB}. 
For each edge~$e$ let $P[e]$ be the polyline of $P$ that is enclosed by $e$.  
We call the polygon describing the area enclosed by $e$ and $P[e]$ the \emph{pocket} of $e$; see Figure~\ref{fig:intro:inputBC}. 
We direct $e$ of $G_{\mathcal C}$ such that it starts at the starting point of $P[e]$ and ends at the endpoint of $P[e]$. 
For each edge $e$ we introduce costs that rate the length $\length(e)$ of $e$ as well as the area $\area(P[e])$ of the pocket of $e$ with respect to $\lambda$, i.e. 
$ \cost(e) = \lambda \cdot c_L(e)+ (1-\lambda) \cdot \area(P[e])$.

\begin{obs}\label{lemma:convexHull}
The vertices of the convex hull of $P$ are part of the boundary of any shortcut hull of $P$.
\end{obs}
Due to Observation~\ref{lemma:convexHull}, any $\mathcal C$-hull of $P$ contains the topmost vertex $v$ of $P$.
Hence, $G_{\mathcal C}$ does not contain any edge $e$ that contains $v$ in its pocket and when removing $v$ from $G_{\mathcal C}$ we obtain a directed acyclic graph. 
We use this property to prove that a cost-minimal path in $G_{\mathcal C}$ corresponds to an optimal $\mathcal C$-hull. 
\begin{restatable}{thm}{thmShortestPath}\label{theorem:shortest-path}
 The problem \textsc{ShortcutHull} without holes can be solved 
 in $O(|\mathcal C|)$ time. In particular, in the case that the edges in $\mathcal C$ do not cross each other it can be solved in $O(n)$ time and  $O(n^2)$ time otherwise. 
\end{restatable}
\begin{proof}
    \begin{figure}[tb]
        \centering
        \includegraphics{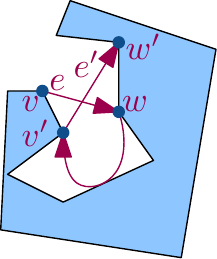}
        \caption{Illustration of proof for Theorem~\ref{theorem:shortest-path}. Due to the order of the vertices of $P$ the edges $e$ and $e'$ cannot both be part of~$S$.}
        \label{fig:shortest-path-crossing-free}
    \end{figure}
    Let $S =(e_1,\hdots,e_l)$ be the sequence of edges of the shortest path in $G_{\mathcal C}$ starting at $v$ and ending at $v$. 
    Let $Q$ be the polygon that we obtain by interpreting
    $S$ as a polygon. 
    We show that $Q$ is an optimal $\mathcal C$-hull. 
    In particular, we need to show that $Q$ is crossing-free.
    Due to the definition of $G_{\mathcal C}$, the following two properties hold: (i) each edge $e=xy$ of $G_{\mathcal C}$ starts and ends on the boundary of $P$ and (ii) $e$ is directed such that $x$ is the starting point of $P[e]$ and $y$ is the end point of $P[e]$.
    Hence, the vertex $x$ appears before $y$ on the boundary of $P$ when going along $P$ starting at its topmost point.
    Assume that the edges $e_i=x_iy_i$ and $e_j=x_jy_j$ with $1\leq i<j\leq l$ cross. 
    Since $i<j$, the start and end points of $e_i$ and $e_j$ appear in the order $x_iy_ix_jy_j$ on $S$.
    Due to properties (i) and (ii), $x_j$ lies in the pocket of $e_i$. 
    Let $S_{i,j} = (e_{i+1},\hdots, e_{j-1})$.
    Since properties (i) and (ii) apply for each edge in $S_{i,j}$, this is a contradiction. 
    The computation of a 
    shortest path in a directed acyclic graph with $|\mathcal C|$ vertices and edges takes 
    $O(|\mathcal C|)$ time~\cite{cormen2009introduction}. In particular, when no two edges of $\mathcal C$ cross, we obtain $O(n)$ running time and otherwise $O(n^2)$.  
\end{proof}
If we allow $Q$ to have holes, we cannot rate the costs for the area of a pocket in advance.

\section{Structural Results for Shortcut Hulls with Holes}\label{sec:structural-results}
   \begin{figure}[t]
    \centering
    
  \begin{subfigure}[t]{0.32\linewidth}
    \includegraphics[page=5,width=\textwidth]{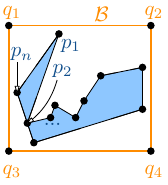}
    \caption{$\mathcal B$}
    \label{fig:intro:ourContrB}
  \end{subfigure}
  \begin{subfigure}[t]{0.32\linewidth}
    \includegraphics[page=6,width=\textwidth]{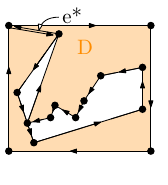}
    \caption{$D$}
    \label{fig:intro:ourContrC}
  \end{subfigure}
  \begin{subfigure}[t]{0.32\linewidth}
     \includegraphics[page=7,width=\textwidth]{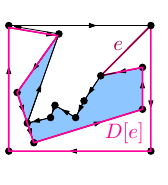}
     \caption{pocket of $e$}
     \label{fig:intro:ourContrD}
    \end{subfigure}
   \caption{Containing box~$\mathcal B$ and  \donutPoly~$D$ of $P$.}
   \label{fig:intro:ourContr}
\end{figure}
In this section, we present structural results for \textsc{ShortcutHull}, which we utilize for an algorithm in Section~\ref{sec:basic-algorithm}. We allow the shortcut hull to have holes.

\subsection{Basic Concepts}
Let $P$ be a weakly-simple polygon with connected exterior. Let $p_1,\hdots, p_n$ be the vertices of $P$; see  Figure~\ref{fig:intro:ourContrB}. We assume that the topmost vertex of $P$ is uniquely defined; we always can rotate $P$ such that this is the case. We denote that vertex by $p_1$ and assume that $P$ is clockwise oriented. Further, let $\mathcal{C}$ be a set of shortcuts of $P$ and $\lambda \in [0,1]$; see Figure~\ref{fig:intro:inputA}.
Due to Observation~\ref{lemma:convexHull}, any $\mathcal C$-hull of $P$ contains~$p_1$.

First we introduce concepts for the description of the structural results and the algorithm.
Let $\mathcal B$ be an axis-aligned rectangle such that it is  slightly larger than the bounding box of $P$; see Figure~\ref{fig:intro:ourContrB}.
Let $q_1,\hdots,q_4$ be the vertices of $\mathcal B$ in clockwise order such that $q_1$ is the top-left corner of $\mathcal{B}$. We require that the diagonal edges $q_1q_3$ and $q_2q_4$ intersect $P$, which is always possible.   
We call $\mathcal B$ a \emph{\containingPoly} of $P$.
Let $D$ be the polygon $q_1\hdots q_4q_1p_1p_n\hdots p_1 q_1$.
We call $D$ a \emph{\donutPoly} of $P$; see Figure~\ref{fig:intro:ourContrC}. We observe that $D$ is a weakly-simply polygon whose interior is one connected region. Further, we call $e^\star=p_1q_1$ the \emph{cut edge} of $D$. 
For an edge $e$ in the interior of $D$ connecting two vertices of $D$ let $D[e]$ be the polyline of $D$ that connects the same vertices such that $e^\star$ is not contained; see Figure~\ref{fig:intro:ourContrD}. Let $D[e]+e$ be the polygon that we obtain by concatenating $D[e]$ and $e$ such that $e^\star$ lies in the exterior of $D[e]+e$.
Note that if $e \in \mathcal C$ then $D[e] = P[e]$.
We call $D[e]+e$ the \emph{pocket} $e$. In particular, we define $D$ to be the pocket of~$e^\star$.

\begin{obs}
The edges  of a $\mathcal C$-hull of $P$ are contained in the \donutPoly $D$.
\end{obs}

In the following, we define a set $\mathcal C^+$ of edges  in $D$ with $\mathcal C\subseteq \mathcal C^+$  that we use for constructing triangulations of $D$, which encode the shortcut hulls. Generally, a \emph{triangulation} of a polygon~$H$ is a superset of the edges of $H$ such that they partition the interior of $H$ into triangles. Further, for a given set $E$ of edges an $E$-\emph{triangulation} of $H$ is a triangulation of $H$ that only consists of edges from $E$.
Moreover, we say that a set~$E$ of edges is \emph{part of} a triangulation $T$ if $E$ is a subset of the edges of $T$. Conversely, we also say that $T$ \emph{contains} $E$ if $E$ is part of $T$.
Note that the edges of $H$ are part of any $E$-triangulation of~$H$.
 
We call a set $\mathcal C^+$ of edges with $\mathcal C\subseteq \mathcal C^+$ an \emph{\tighthulledges} of the shortcuts $\mathcal C$ and the \donutPoly $D$ if 
\begin{inparaenum}[(1)]
  \item every edge of $\mathcal C^+$ is contained in $D$,
  \item every edge of $\mathcal C^+$ starts and ends at vertices of $D$, and
  \item for every set $\mathcal C'\subseteq \mathcal C$ of pair-wisely non-crossing edges there is a $\mathcal C^+$-triangulation~$T$ of $D$ such that $\mathcal C'$ is part of $T$.
\end{inparaenum}
First, we observe that $\mathcal C^+$ is well-defined as every edge in~$\mathcal C$ satisfies the first two properties. Further, by definition for any $\mathcal C$-hull $Q$ there is a $\mathcal C^+$-triangulation~$T$ of $D$ that contains $Q$. Hence, as an intermediate step our algorithm for computing an optimal $\mathcal C$-hull~$Q$ creates an \tighthulledges of $\mathcal C$ and $D$, and then constructs a $\mathcal C^+$-triangulation that contains $Q$. In Section~\ref{sec:Cistighthulledge} we discuss the structural correspondences between $\mathcal C^+$-triangulations of $D$ and (optimal) $\mathcal C$-hulls. In Section~\ref{sec:extendedEdgeSet} we then show how to construct $\mathcal C^+$.
For example a simple approach for an \extendedEdgeSet of $\mathcal C$ is the set of all possible shortcuts in $D$.
We observe that 
any enrichment $\mathcal C^+$ of $\mathcal C$ has $O(n^2)$ edges. 
In general, the size of $\mathcal C^+$ can be described by the spatial complexity of $\mathcal C$, which impacts the running time of our algorithm (Section~\ref{sec:basic-algorithm}).

\subsection{From $\mathcal C^+$-Triangulations to $\mathcal C$-Hulls}\label{sec:Cistighthulledge}
In this section, we assume that we are given an \extendedEdgeSet 
$\mathcal C^+$ for the set of shortcuts $\mathcal C$ and a 
\donutPoly $D$. 
Let $T$ be a $\mathcal C^+$-triangulation of $D$; see Figure~\ref{fig:cplus-triangulation}.

\begin{figure}
    \centering
    \begin{subfigure}[t]{0.32\linewidth}
    \includegraphics[page=1,width=\textwidth]{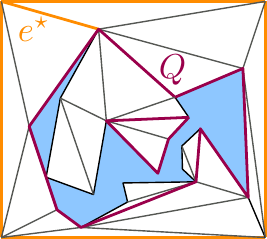}
    \caption{$Q$ is part of $T$}
    \label{fig:cplus-triangulationA}
    \end{subfigure}
    \begin{subfigure}[t]{0.32\linewidth}
    \includegraphics[page=2,width=\textwidth]{figs/cplus-triangulation.pdf}    \caption{labels}
    \label{fig:cplus-triangulationB}
    \end{subfigure}
    \begin{subfigure}[t]{0.32\linewidth}
    \includegraphics[page=3,width=\textwidth]{figs/cplus-triangulation.pdf}
    \caption{dual graph~$G^\star$}
    \label{fig:cplus-triangulationC}
    \end{subfigure}
    \caption{$\mathcal C^+$-triangulation~$T$. (b)~The red triangles are active, while all other triangles are inactive. (c)~The restricted dual graph~$G^\star$ of $T$ forms a tree with root $\rho$. }
    \label{fig:cplus-triangulation}
\end{figure}

\begin{obs}\label{lem:tight-hull-in-triangulation}
For each \extendedEdgeSet $\mathcal C^+$ of $\mathcal C$ and each $\mathcal C$-hull 
$Q$ there exists a $\mathcal C^+$-triangulation $T$ of the \donutPoly 
$D$ such that $Q$ is part of $T$. 
\end{obs}

Let $T$ be a $\mathcal C^+$-triangulation of $D$ such that the $\mathcal C$-hull $Q$ is part of $T$; see Figure~\ref{fig:cplus-triangulationA}. 
We can partition the set of triangles of $T$ in those that are contained in the interior of $Q$ and those that are contained in the exterior of $Q$. We call the former ones \emph{active} and the latter ones \emph{inactive}; see Figure~\ref{fig:cplus-triangulationB}. 
Further, we call an edge~$e$ of $T$ a \emph{separator} if 
\begin{inparaenum}[(1)]
    \item it is part of $P$ and adjacent to an inactive triangle, or
   \item it is adjacent to both an active and an inactive triangle.
\end{inparaenum}
Conversely, let $\ell\colon T \to \{0,1\}$ be a labeling of $T$ that assigns to each triangle~$\Delta$ of $T$ whether it is active ($\ell(\Delta)=1$) or inactive ($\ell(\Delta)=0$). We call the pair $\mathbf T = (T,\ell)$ a \emph{labeled} $\mathcal C^+$-triangulation. 
From Observation~\ref{lem:tight-hull-in-triangulation} we obtain the next observation.

\begin{obs}\label{lem:tight-hull-labeling}
For each \tighthulledges $\mathcal C^+$ of $\mathcal C$ and each $\mathcal C$-hull $Q$ there exists a labeled $\mathcal C^+$-triangulation such that its separators stem from $\mathcal C$ and form $Q$.
\end{obs}

Let $\mathbf T=(T,\ell)$ be a labeled $\mathcal C^+$-triangulation of the interior of a polygon~$H$. We denote the set of separators of $\mathbf T$ by $S_\mathbf T$. We define  
\[\length(S_{\mathbf T})=\sum_{e\in S_\mathbf T}  \length(e) \text{ and } \area(T)=\sum_{\substack{\Delta\in T,\\\ell(\Delta)=1}} \area(\Delta),\]
where $\length(e)$ denotes the length of $e$ and $\area(\Delta)$ denotes the area of $\Delta$.
The \emph{costs} of $\mathbf T$ are then defined as
\[
  \cost(\mathbf T)=\lambda\cdot \length(S_{\mathbf T}) + (1-\lambda)\cdot \area(T).
\]
For any $e \in \mathcal C^+ \setminus \mathcal C$ we define $\length(e)=\infty$.
Thus, we have $\cost(\mathbf T)<\infty$ if and only if $S_\mathbf T\subseteq \mathcal C$.
We call a labeled $\mathcal C^+$-triangulation~$\mathbf T$ of~$H$ \emph{optimal} if there is no other labeled $\mathcal C^+$-triangulation~$\mathbf T'$ of $H$ with $\cost(\mathbf T')<\cost(\mathbf T)$. 

Next, we show that a labeled $\mathcal C^+$-triangulation~$\mathbf T=(T,\ell)$ that is optimal can be recursively constructed based on optimal sub-triangulations. 
Let $G^\star$ be the restricted dual graph of $T$, i.e., for each triangle $G^\star$ has a node and two nodes are adjacent iff the corresponding triangles are adjacent in $T$; see Figure~\ref{fig:cplus-triangulationC}.

\begin{restatable}{lem}{lemmaDecompTree}\label{lemma:decompositionTree}
The restricted dual graph $G^\star$ of a $\mathcal C^+$-triangulation $T$ of $D$ is a binary tree. 
\end{restatable}

\begin{proof}
  As each edge of $T$ starts and ends at the boundary of $D$, each edge of $T$ splits $D$ into two disjoint regions. Hence, $G^\star$ is a tree.
Further, since each node of $G^\star$ corresponds to a triangle of $T$, each node of $G^\star$ has at most two child nodes.
\end{proof}

We call $G^\star$ a \emph{decomposition tree} of $D$.
Let $\rho$ be the node of $G^\star$ that corresponds to the triangle of $T$ that is adjacent to the cut edge $e^\star$ of $D$; as $e^\star$ is a boundary edge of $D$, this triangle is uniquely defined. We assume that $\rho$ is the root of $G^\star$; see Figure~\ref{fig:cplus-triangulationC}. Let $G^\star_u$ be an arbitrary sub-tree of $G^\star$ that is rooted at a node $u$ of $G^\star$. Further, let $e_u$ be the edge of the triangle~$\Delta_u$ of $u$ that is not adjacent to the triangles of the child nodes of $u$; we call $e_u$ the \emph{base edge} of $\Delta_u$. The triangles of the nodes of $G^\star_u$ form a $\mathcal C^+$-triangulation~$T_u$ of the pocket $A_u=D[e_u]+e_u$ of $e_u$. Thus, $G^\star_u$ is a decomposition tree of $A_u$.
A labeled $\mathcal C^+$-sub-triangulation $\mathbf T_u=(T_u,\ell_u)$ consists of the $\mathcal C^+$-triangulation $T_u$ of $A_u$ with $T_u\subseteq T$ and the labeling $\ell_u$ with $\ell_u(\Delta)=\ell(\Delta)$ for every $\Delta\in T_u$.

\begin{restatable}{lem}{lemmaOptSubtree}\label{lem:recursive-decomposition}
 Let $\mathbf T$ be a labeled $\mathcal C^+$-triangulation of $D$ that is optimal. 
 Let $\mathbf T_u=(T_u,\ell_u)$ be the labeled $\mathcal C^+$-sub-triangulation of $\mathbf T$ rooted at the node $u$ and let $\mathbf T'_u=(T'_u,\ell'_u)$ be an arbitrary labeled $\mathcal C^+$-triangulation of the same region. We denote the triangles of $\mathbf T_u$ and $\mathbf T'_u$ adjacent to $e_u$ by $\Delta_u$ and $\Delta'_u$, respectively.  
 
 If $\Delta_u$ and $\Delta'_u$ have the same labels, i.e., $\ell_u(\Delta_u)=\ell'_u(\Delta'_u)$, then $\cost(\mathbf T_u) \leq \cost(\mathbf T'_u)$.
\end{restatable}

\begin{proof}
  For the proof we use a simple exchange argument. Assume that there is a labeled $\mathcal C^+$-triangulation $\mathbf T'_u$ of the pocket $D[e_u]+e_u$ with $\ell(\Delta_u)=\ell'(\Delta'_u)$ and $\cost(\mathbf T'_u)<\cost(\mathbf T_u)$.
  As both $\mathbf T_u$ and $\mathbf T_u'$ are triangulations of the pocket $D[e_u]+e_u$, we can replace the triangles of $T_u$ with the triangles of $T'_u$ in $T$ obtaining a new triangulation ${T'}$ of $D$. 
  Further, we define a new labeling ${\ell}'$ such that  ${\ell}'(\Delta)=\ell(\Delta)$ for every $\Delta \in T\setminus T_u$ and ${\ell}'(\Delta)=\ell(\Delta)$ for every $\Delta \in T'_u$. Let ${\mathbf T'}=({T}',{\ell}')$ be the corresponding labeled $\mathcal C^+$-triangulation of $D$. 
  The following calculation shows $\cost({\mathbf T'})<\cost(\mathbf T)$, which contradicts the optimality of $\mathbf T$. 
  \begin{align*}
   \cost({\mathbf {T'}})=\,& \lambda\cdot \left(\length(S_\mathbf{{T'}}\setminus S_{\mathbf T'_u}) +  \length(S_{\mathbf T'_u})\right)+
   \\ &\, (1-\lambda)\cdot\left( \area({T'}\setminus T'_u) + \area(T'_u)\right)\\
    = &\, \lambda\cdot \length(S_\mathbf{T}\setminus S_{\mathbf T'_u}) + (1-\lambda)\cdot \area({T}\setminus T'_u)+\\
    &\, \lambda\cdot \length(S_{\mathbf T'_u}) + (1-\lambda)\cdot \area(T'_u) \\
    = &\, \lambda\cdot \length(S_\mathbf{T}\setminus S_{\mathbf T'_u}) +(1-\lambda)\cdot \area({T}\setminus T'_u)+ \cost(\mathbf T'_u)\\
    < &\, \lambda\cdot \length(S_\mathbf{T}\setminus S_{\mathbf T_u}) + \cost(\mathbf T_u) = \cost(\mathbf T)
  \end{align*}
  Altogether, we obtain the statement of the lemma.
\end{proof}

We use Lemma~\ref{lem:recursive-decomposition} for a dynamic programming approach that yields a labeled $\mathcal C^+$-triangulation $\mathbf T$ of $D$ that is optimal.

\begin{lem}\label{lem:from-triangulation-to-hull}
 Let $\mathbf T$ be a labeled $\mathcal C^+$-triangulation of $D$ that is optimal and has cost $\cost(\mathbf T)<\infty$. The separators of $\mathbf T$ form an optimal $\mathcal C$-hull of $P$.
\end{lem}
\begin{proof}
    We show the following two claims, which proves the lemma.
    \begin{inparaenum}[(1)]
    \item For every $\mathcal C$-hull $Q$ of $P$ there is a labeled $\mathcal C^+$-triangulation~$\mathbf T$ of $D$ such that the separators of $\mathbf T$ form $Q$ and $\cost(\mathbf T)=\cost(Q)$. 
  \item For every labeled $\mathcal C^+$-triangulation~$\mathbf T$ of $D$ with $\cost(\mathbf T)<\infty$ the separators of $\mathbf T$ form a $\mathcal C$-hull~$Q$ with $\cost(\mathbf T)=\cost(Q)$.
 \end{inparaenum}
 
     \textit{Claim 1.} Let $Q$ be a $\mathcal C$-hull of $P$. By the definition of $\mathcal C^+$ there is a $\mathcal C^+$-triangulation~$T$ of $D$ such that $Q$ is part of $T$. We define the labeling $\ell$ such that $\ell(\Delta)=1$ for every triangle~$\Delta \in T$ that is contained in the interior of $Q$ and $\ell(\Delta)=0$ for every other triangle $\Delta \in T$. Hence, the separators of the labeled $\mathcal C^+$-triangulation $\mathbf T=(T,\ell)$ are the edges of $Q$. Further, by the construction of $\mathbf T$ we have $\cost(\mathbf T)=\cost(Q)$. This proves Claim~1.

    \begin{figure}
        \centering
        \begin{subfigure}[t]{0.45\linewidth}
            \includegraphics[page=1,width=\textwidth]{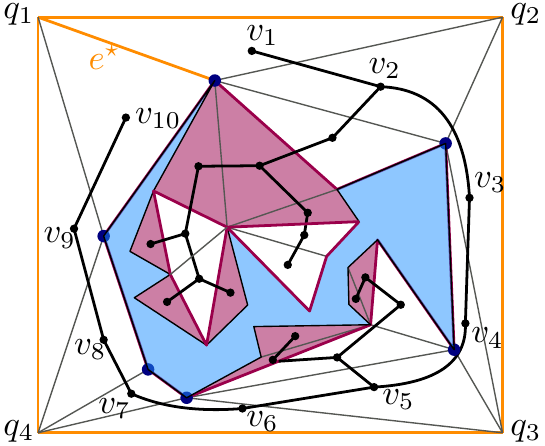}
            \caption{}
            \label{fig:construction-of-Q:setting}
        \end{subfigure}
        \begin{subfigure}[t]{0.45\linewidth}
            \includegraphics[page=2,width=\textwidth]{figs/proof-cplus-triangulation-to-hull.pdf}
            \caption{}
            \label{fig:construction-of-Q:outer-hull}
        \end{subfigure}
        \caption{Proof of Lemma~\ref{lem:from-triangulation-to-hull}. (a)~The triangles incident to the vertices $q_1$, $q_2$, $q_3$ and $q_4$ form a path in the dual graph of the labeled triangulation~$\mathbf T$. (b)~The vertices $p_1,\dots,p_5$ form a $\mathcal C^+$-hull of $P$ containing all active triangles (red) of $\mathbf T$. }
        \label{fig:construction-of-Q}
    \end{figure}
    \begin{figure}
        \centering
        \begin{subfigure}[t]{0.32\linewidth}
            \includegraphics[page=3,width=\textwidth]{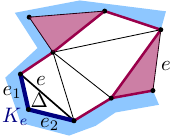}
            \caption{}
            \label{fig:construction-of-Q:induction:base-case}
        \end{subfigure}
        \begin{subfigure}[t]{0.32\linewidth}
            \includegraphics[page=4,width=\textwidth]{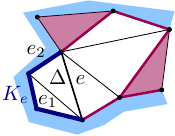}
            \caption{}
            \label{fig:construction-of-Q:induction:caseA}
        \end{subfigure}
        \begin{subfigure}[t]{0.32\linewidth}
            \includegraphics[page=6,width=\textwidth]{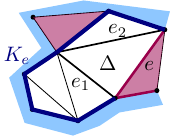}
            \caption{}
            \label{fig:construction-of-Q:induction:caseB}
        \end{subfigure}
        \caption{Inductive construction of the boundary path~$K_e$ of an edge $e$ that is a base edge of an inactive triangle $\Delta$. (a)~Base case. (b)~$e_1$ is a base edge of an inactive triangle, and $e_2$ is a separator. (c)~Both $e_1$ and $e_2$ are base edges of inactive triangles.}
        \label{fig:construction-of-Q:induction}
    \end{figure}
    
    \textit{Claim 2.} Let $\mathbf T=(T,\ell)$ be a $\mathcal C^+$-triangulation of $D$ with $\cost(\mathbf T)<\infty$ and let $S_{\mathbf T}$ be the separators of $\mathbf T$. By the definition of the costs of $\mathbf T$ we have $S_{\mathbf T} \subseteq \mathcal C$. Moreover, as $T$ is a triangulation, the edges in $S_{\mathbf T}$ do not cross each other. We show that the edges in $S_{\mathbf T}$ form a $\mathcal C$-hull~$Q$ with $\cost(Q)=\cost(\mathbf T)$. Let $G^\star$ be the dual graph of $T$. As the diagonal edges of the \containingPoly  $\mathcal B$ intersect $P$, each triangle of~$T$ that is incident to one of the vertices of $\mathcal B$ is also incident to a vertex of $P$; see Figure~\ref{fig:construction-of-Q:setting}.
    The vertices of the triangles incident to the vertices of $\mathcal B$ form a path $v_1,\dots,v_k$ in $G^\star$ such $v_1$ is the root of $G^\star$ and $v_k$ is a leaf. We denote the triangles represented by this path by $\Delta_1,\dots,\Delta_k$, respectively.

    Let $p_1,\dots,p_l$ be the vertices of $P$ in the order as they are incident to the triangles $\Delta_1,\dots,\Delta_k$ in clockwise order; see Figure~\ref{fig:construction-of-Q:setting}. We define $p_{l+1}=p_1$. The vertices $p_1,\dots,p_l$ form a weakly-simple polygon~$Q'$ that contains $P$; if $P$ crossed $Q'$, this would contradict that the vertices are incident to the disjoint triangles $\Delta_1,\dots,\Delta_k$. We observe that $Q'$ is a $\mathcal C^+$-hull of $P$ without holes. Let $T'\subseteq T$ be the set of triangles that are contained in $Q'$ and let $E'$ be the edges of these triangles. We first show that for each edge $e \in E'$ that is a base edge of an inactive triangle in $\mathbf T$ there is a path $K_e$ in the pocket of $e$ such that
    \begin{inparaenum}[(1)]
       \item $K_e$ only consists of edges from $S_{\mathbf T}$,
       \item $K_e$ connects the endpoints of $e$, and 
       \item the polygon $K_e+e$ only contains inactive triangles of $\mathbf T$. 
    \end{inparaenum}
    We call $K_e$ the \emph{boundary path} of $e$; see Figure~\ref{fig:construction-of-Q:induction}. Later, we use these boundary paths to assemble $Q$.

     Let $\Delta$ be the inactive triangle of which $e$ is the base edge and let $e_1$ and $e_2$ be the other two edges of $\Delta$. We do an induction over the number of triangles of $\mathbf T$  that are contained in the pocket of $e$.
     If the pocket of $e$ only contains $\Delta$, both edges $e_1$ and $e_2$ are edges of $P$; see Figure~\ref{fig:construction-of-Q:induction:base-case}. Hence, by definition they are separators. We define 
    $K_e$ as the path $e_1+e_2$, which satisfies the three requirements above. So assume that the pocket of $e$ contains more than one triangle; see Figure~\ref{fig:construction-of-Q:induction:caseA}--\subref{fig:construction-of-Q:induction:caseB}. If $e_1$ is not a separator, then it is the base edge of an inactive triangle. Hence, by induction there is a path $K_{e_1}$ that satisfies the requirements above. If $e_1$ is a separator, we define $K_{e_1}=e_1$. In the same way we define a path $K_{e_2}$ for the edge $e_2$. The concatenation $K_{e_1}+K_{e_2}$ forms a path that satisfies the requirements above, which proves the existence of the boundary path for an edge $e\in E'$.  
    
    We now describe the construction of the boundary of $Q$. For a pair $p_i,p_{i+1}$ with $1\leq i < l$ the adjacent triangle incident to one of the vertices of $\mathcal B$ is inactive. Let $K_i=p_ip_{i+1}$ if $p_ip_{i+1}$ is a separator.  Otherwise, $p_ip_{i+1}$ is the base edge of an inactive triangle in $\mathbf T$. Thus, it has a boundary path $K_{p_ip_{i+1}}$ and we define $K_i$ as $K_e$. The concatenation $K_1+\cdots+K_{l}$ forms the boundary~$B$ of a weakly-simple polygon~$Q$ that encloses $P$; see Figure~\ref{fig:construction-of-Q:outer-hull}. By construction it consists of edges from~$\mathcal C$. 
    
    Finally, we show how to construct the holes of $Q$. Let $e\in S_{\mathbf T}$ be a separator that is contained in the interior of $B$ and that is a base edge of an inactive triangle; see $e$ and $e'$ in Figure~\ref{fig:construction-of-Q:outer-hull}. The polygon $Z_e$ that consists of $e$ and the boundary path $K_e$ only contains inactive triangles of $\mathbf T$ and is entirely contained in $B$. Further, for any pair $e$ and $e'$ of such separators in the interior of $B$ the interiors of the polygons $Z_{e}$ and $Z_{e'}$ are disjoint. Hence, we set these polygons to be the holes of $Q$. Thus, we obtain a $\mathcal C$-hull~$Q$ of $P$ with holes such that the inactive triangles of $\mathbf T$ lie in the exterior of $Q$, while all active triangles lie in the interior of $Q$. This implies that  $\cost(Q)=\cost(\mathbf{T})$, which concludes the proof of Claim~2. 
\end{proof}

\subsection{From $\mathcal C$ to  $\mathcal C^+$}\label{sec:extendedEdgeSet}
\begin{figure}
    \centering
    \begin{subfigure}[t]{0.45\linewidth}
    \includegraphics[page=1,width=\textwidth]{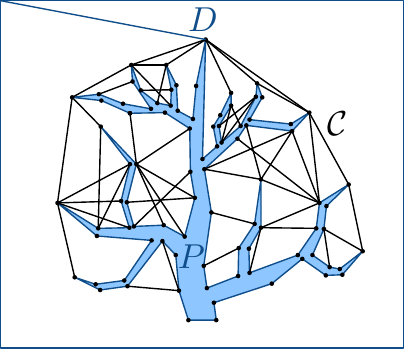}
    \caption{$P$, $\mathcal C$, and $D$}
    \label{fig:intersected-trianglesA}
    \end{subfigure}
    \begin{subfigure}[t]{0.45\linewidth}
    \includegraphics[page=2,width=\textwidth]{figs/intersectingregions.pdf}
    \caption{crossing components}
    \label{fig:intersected-trianglesB}
    \end{subfigure}
    
    \begin{subfigure}[t]{0.45\linewidth}
    \includegraphics[page=3,width=\textwidth]{figs/intersectingregions.pdf}
    \caption{edges $E_T$}
    \label{fig:intersected-trianglesC}
    \end{subfigure}
    \begin{subfigure}[t]{0.45\linewidth}
    \includegraphics[page=4,width=\textwidth]{figs/intersectingregions.pdf}
    \caption{$\mathcal C^+$}
    \label{fig:intersected-trianglesD}
    \end{subfigure}
    \caption{Obtaining the enrichment $\mathcal C^+$ from $\mathcal C$.}
    \label{fig:intersected-triangles}
\end{figure}

Solving \textsc{ShortcutHull} relies on the considered enrichment $\mathcal C^+$. 
For an edge $e \in \mathcal C^+$ let $\delta_e$ be the number of triangles that can be formed by $e$ and two other edges from $\mathcal C^+$, and let $\delta(\mathcal C^+)$ be the maximum $\delta_e$ over all edges $e$ in $\mathcal C^+$.
In Section~\ref{sec:basic-algorithm} we show that the problem can be solved in $O(|\mathcal C^+|\cdot \delta(\mathcal C^+))$ time.

A simple choice for $\mathcal C^+$ is the set of all edges that lie in $D$ and connect vertices of $D$. It is an enrichment of $\mathcal C$ as it contains any choice of $\mathcal C$ and any triangulation of $D$ that is based on the vertices of $D$ is a subset of $\mathcal C^+$.
\begin{obs}
 There is an  \tighthulledges $\mathcal C^+$ of $\mathcal C$ with $|\mathcal C^+| \in O(n^2)$ and $\delta(\mathcal C^+)\in O(n)$.
\end{obs}
If $\mathcal C$ has no crossings, we can do much better. We first observe that the edges of any triangulation~$T$ of the \donutPoly $D$ are an \tighthulledges of $\mathcal C$ and $D$ if $\mathcal C$ is a subset of these edges. Hence, we can define an \tighthulledges as the set of edges of a triangulation~$T$ of $D$ such that the edges of $\mathcal C$ are part of $T$; for this purpose we can for example utilize constrained Delaunay triangulations, but also other triangulations are possible. 
\begin{obs}
 If the edges in $\mathcal C$ do not cross, 
 $\mathcal C$ has an \tighthulledges $\mathcal C^+$ with $|\mathcal C^+|\in O(n)$ and $\delta(\mathcal C^+)\in O(1)$.
\end{obs}

In the following we generalize both constructions of $\mathcal C^+$ and relate $|\mathcal C^+|$ and $\delta(\mathcal C^+)$ to the number $n$ of vertices of $P$ and the spatial complexity $\chi$ of $\mathcal C$.   
Let $\mathcal C_1,\hdots, \mathcal C_h$ be subsets of $\mathcal C$ such that two edges $e \in \mathcal C_i$ and $e' \in \mathcal C_j$ with 
$1\leq i,j \leq h$ cross each other if and only if $i = j$; see Figure~\ref{fig:intersected-triangles}. 
We call $\mathcal C_i$ a \emph{crossing component} of $\mathcal C$. Let $R_i$ be the polygon in~$D$ with fewest edges, that is defined by vertices of $P$ and contains $C_i$.
We call $R_i$ the \emph{region} of $C_i$.
Let $\mathcal C^+$ be the set of edges that contains (i) all edges of $\mathcal C$, (ii) the edges $E_T$ of a constrained triangulation for the interior of $D$, and (iii) for each $1\leq i\leq h$ the set $E_{R_i}$ of all possible shortcuts of region $R_i$ such that these start and end at vertices of $R_i$ and are contained in $D$.
Hence, an \tighthulledges is of size $O(\chi^2 + n)$ as each region $R_i$ has at most $\chi$ vertices. 

\begin{thm}\label{thm:cplus:high-number-of-crossings}
 There is an  \tighthulledges~$\mathcal C^+$ of $\mathcal C$ with $|\mathcal C^+| \in O(\chi^2+n)$ and $\delta(\mathcal C^+)\in O(\chi)$.
\end{thm}
\begin{proof}
    Let $\mathcal C^+$ be the set of edges that contains all edges of $\mathcal C$, $E_T$, and $E_{R_1},\hdots,E_{R_h}$. 
    We show that $\mathcal C^+$ is an~\tighthulledges, by proving that  for each set $\mathcal C' \subseteq \mathcal C$ of pair-wisely non-crossing edges there is a $\mathcal C^+$-triangulation $T$ of $D$ such that $\mathcal C'$ is part of $T$.

    Observe that the regions $R_1,\hdots,R_h$ of crossing components induce a partition $\mathbf{R}$ of $D$ that contains $R_1,\hdots, R_h$ and regions   $R'_1,\hdots, R'_g$ partitioning $D\setminus\bigcup_{i=1}^h R_i$.
    Since an edge $e \in \mathcal C^+$ cannot cross the boundary of two regions $R,R' \in \mathbf{R}$, the triangulation of each region $R \in \mathbf{R}$ can be constructed independently. 
    
    Let $E$ be the edges of $\mathcal C'$ that are contained in region $R \in \mathbf{R}$. 
    If $R$ is a region of a crossing component, $\mathcal C^+$ contains all shortcuts in this region. 
    Since the edges of $E$ are crossing-free, there exists a $\mathcal C^+$-triangulation of $R$ that is constrained to $E$. 
    Thus, the edges of $E$ are part of a $\mathcal C^+$-triangulation of $R$. 
    If $R$ is not a region of a crossing component, the \tighthulledges~$\mathcal C^+$ contains the edges of a triangulation of $D$ constrained to all edges of $\mathcal C$ that are contained in $R$.
    Since $E \subseteq \mathcal C$, this triangulation contains all edges of $E$. 
    By joining the $\mathcal C^+$-triangulations for each region of the partition, we obtain a $\mathcal C^+$-triangulation of $D$ such that $\mathcal C'$ is part of it. 
\end{proof}

\section{Computing Optimal Shortcut Hulls with Holes}\label{sec:basic-algorithm}
The core of our algorithm is a dynamic programming approach that recursively builds the decomposition tree of $T$ as well as the labeling $\ell$ using the \donutPoly~$D$ of the input polygon~$P$ and the input set of shortcut~$\mathcal C$ as guidance utilizing Lemma~\ref{lem:recursive-decomposition}.   
The algorithm consists of the following steps. 
\begin{compactenum}
  \item Create a containing box $\mathcal B$ and the \donutPoly~$D$ of $P$ and $\mathcal B$. Let $e^\star$ be the cut edge of $D$.
  \item Create an \extendedEdgeSet $\mathcal C^+$ of $\mathcal C$ and $D$.
  \item Create the geometric graph $G_{\mathcal C^+}$ based on $\mathcal C^+$. Let $\mathcal T$ be the set of triangles in $G_{\mathcal C^+}$.
  \item Determine for each edge~$e$ of $G_{\mathcal C^+}$ the set $T_e \subseteq T$ of all triangles $(e,e_1,e_2)$ in $G_{\mathcal C^+}$  
  such that $e_1$ and $e_2$ lie in the pocket of $e$.
  \item Create two tables $\mathrm A$ and $\mathrm I$ such that they have an entry for each edge~$e$ of $G_{\mathcal C^+}$.\label{algorithm:step-table} 
    \begin{compactitem}
        \item $\mathrm A[e]$: minimal cost of a labeled $\mathcal C^+$-tri\-an\-gulation~$\mathbf T$ of the pocket $D[e]$+$e$ s.t.\  the triangle adjacent to $e$ is active.
        \item $\mathrm{I}[e]$: minimal cost of a labeled $\mathcal C^+$-tri\-an\-gulation~$\mathbf T$ of the pocket $D[e]$+$e$ s.t.\  the triangle adjacent to $e$ is inactive.
    \end{compactitem}
   \item Starting at $\mathrm I[e^\star]$ apply a backtracking procedure to create a $\mathcal C^+$-triangulation~$\mathbf T$ of $D$ that is optimal.
   Return $\mathbf T$ and the corresponding optimal $\mathcal C$-hull~$Q$ of $\mathbf T$ (see proof of Lemma~\ref{lem:from-triangulation-to-hull} for construction of $Q$). 
   \label{algorithm:step-backtracking}    
\end{compactenum}
We now explain Step~\ref{algorithm:step-table} and Step~\ref{algorithm:step-backtracking} in greater detail. 

\begin{figure}
  \centering
  \begin{subfigure}[t]{0.24\linewidth}
    \includegraphics[page=1,width=\textwidth]{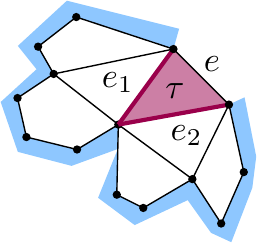}
    \caption{}
  \end{subfigure}
    \begin{subfigure}[t]{0.24\linewidth}
    \includegraphics[page=2,width=\textwidth]{figs/cases.pdf}
    \caption{}
  \end{subfigure}
    \begin{subfigure}[t]{0.24\linewidth}
    \includegraphics[page=3,width=\textwidth]{figs/cases.pdf}
    \caption{}
  \end{subfigure}
      \begin{subfigure}[t]{0.24\linewidth}
    \includegraphics[page=4,width=\textwidth]{figs/cases.pdf}
    \caption{}
  \end{subfigure}

  \begin{subfigure}[t]{0.24\linewidth}
    \includegraphics[page=1,width=\textwidth]{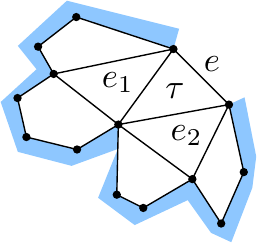}
    \caption{}
  \end{subfigure}
    \begin{subfigure}[t]{0.24\linewidth}
    \includegraphics[page=2,width=\textwidth]{figs/cases-inactive.pdf}
    \caption{}
  \end{subfigure}
      \begin{subfigure}[t]{0.24\linewidth}
    \includegraphics[page=3,width=\textwidth]{figs/cases-inactive.pdf}
    \caption{}
  \end{subfigure}
    \begin{subfigure}[t]{0.24\linewidth}
    \includegraphics[page=4,width=\textwidth]{figs/cases-inactive.pdf}
    \caption{}
  \end{subfigure}
  
  \caption{The possible cases for the (a)--(d) active (red) and (e)--(h) inactive cost of a triangle $\Delta$.}
  \label{fig:basic-algorithm:active-costs}
\end{figure}

\paragraph{Step~\ref{algorithm:step-table}.} We compute the table entries of $\mathrm A$ and $\mathrm I$ in increasing order of the areas of the edges' pockets. Let $e$ be the currently considered edge of $G_{\mathcal C^+}$.
For a triangle $\Delta=(e,e_1,e_2) \in \mathcal T_e$ of $e$ we define its \emph{active cost} $x_\Delta$ as  \[
x_{\Delta} = \sum_{i\in\{1,2\}} \min\{\mathrm{A}[e_i],I[e_i]+\lambda\cdot \length(e_i)\}.  
\] 
Hence, $x_\Delta$ is the cost of a labeled $\mathcal C^+$-triangulation~$\mathbf T_e$ of the pocket $D[e]+e$ such that $\Delta$ is active and the sub-triangulations of $\mathbf T_e$ restricted to the pockets $D[e_1]+e_1$ and $D[e_2]+e_2$ are optimal, respectively; see Figure~\ref{fig:basic-algorithm:active-costs} for the four possible cases. 
\[
   \mathrm A[e] = \begin{cases}
           \infty & \text{$e \not\in \mathcal C$} \\
           \beta\cdot \area(e) & \text{$e\in \mathcal C$, $\mathcal T_e=\emptyset$},\\
           \min\{x_\Delta \mid \Delta \in \mathcal T_e\}+\beta\cdot\area(e) & \text{$e\in \mathcal C$, $\mathcal{T}_e\neq \emptyset$,}\\ 
   \end{cases}
\]
where $\beta=(1-\lambda)$.
Analogously, we define for $\Delta$ its \emph{inactive cost} $y_\Delta$ as
\[
y_\Delta=\sum_{i\in\{1,2\}}\min\{\mathrm{A}[e_i]+\lambda\cdot \length(e_i),I[e_i]\}.  
\]
Hence, $y_\Delta$ is the cost of a labeled $\mathcal C^+$-triangulation~$\mathbf T_e$ of the pocket $D[e]+e$ such that $\Delta$ is inactive and the sub-triangulations of $\mathbf T_e$ restricted to the pockets $D[e_1]+e_1$ and $D[e_2]+e_2$ are optimal, respectively.
We compute the entry $\mathrm I[e]$ as follows.  
\[
   \mathrm I[e] = \begin{cases}
           \infty & \text{$e\in \mathcal C$ and $\mathcal T_e=\emptyset$},\\
           \min\{y_\Delta \mid \Delta \in \mathcal T_e\} & \text{otherwise.}\\ 
   \end{cases}
\]
By the definition of the tables $\mathrm A$ and $\mathrm I$ and Lemma~\ref{lem:recursive-decomposition} it directly follows, that $\mathrm I[e^\star]$ is the cost of a labeled $\mathcal C^+$-triangulation of $D$ that is optimal. In particular, by Lemma~\ref{lem:from-triangulation-to-hull} the entry $\mathrm I[e^\star]$ is the cost of an optimal $\mathcal C$-hull.  

\paragraph{Step~\ref{algorithm:step-backtracking}.} 
When filling both tables, we further store for each entry $\mathrm A[e]$ the triangle $(e,e_1,e_2)\in \mathcal T_e$  with minimum active cost. 
In particular, for the edge $e_i$ (with $i\in \{1,2\}$) we store a pointer to the entry $\mathrm A[e_i]$ if $\mathrm{A}[e_i] < \mathrm I[e_i]+\lambda\cdot \length(e_i)$ and a pointer to the entry $\mathrm I[e_i]$ otherwise.
Similarly, we store for each entry $\mathrm I[e]$ the triangle $(e,e_1,e_2)\in T_e$  with minimum inactive cost. 
In particular, for the edge $e_i$ (with $i\in \{1,2\}$) we store a pointer to the entry $\mathrm I[e_i]$ if $\mathrm{I}[e_i] < \mathrm A[e_i]+\lambda\cdot \length(e_i)$ and a pointer to the entry $\mathrm A[e_i]$ otherwise. 
Starting at the entry $\mathrm I[e^\star]$, we follow the pointers and collect for each encountered entry its triangle
---if such a triangle does not exist, we terminate the traversal. 
If the entry belongs to $\mathrm A$ we label $\Delta$ active and if it belongs to $\mathrm I$, we label $\Delta$ inactive. 
The set~$T$ of collected triangles forms a labeled $\mathcal C^+$-triangulation~$\mathbf T$ of $D$ that is optimal. By Lemma~\ref{lem:from-triangulation-to-hull} the separators of $\mathbf T$ form an optimal $\mathcal C$-hull. 

\paragraph{Running Time.}
The first step clearly runs in $O(n)$ time.
By Theorem~\ref{thm:cplus:high-number-of-crossings}
there is an enrichment $\mathcal C^+$ of $\mathcal C$ and $D$ that has size $O(\chi^2+n)$. It can be easily constructed in $O(\chi^3+\chi n)$ time, which dominates the running times of Step~2, Step~3 and Step~4. 
Further, for each edge $e$ of $G_{\mathcal C^+}$ the set $\mathcal T_e$ contains $\delta(\mathcal C^+)$ triangles. Hence, filling the tables $\mathrm A$ and $\mathrm I$ takes $O(|\mathcal C^+|\cdot \delta(\mathcal C^+))$ time. Hence, by Theorem~\ref{thm:cplus:high-number-of-crossings} we obtain $O(\chi^3+\chi n)$ running time.  The backtracking takes the same time. 

\begin{thm}
 \textsc{ShortcutHull} can be solved in $O(\chi^3+n \chi)$ time. In particular, it is solvable in $O(n^3)$ time in general and in $O(n)$ time if the edges in $\mathcal C$ do not cross.
\end{thm}

\section{Edge and Bend Restricted Shortcut Hulls}\label{sec:variants}
In this section, we discuss two variants of \textsc{ShortcutHull} in which we restrict the number of edges and bends of the computed shortcut hull. These restrictions are particularly interesting for the simplification of geometric objects as they additionally allow us to easily control the complexity of the simplification.

\subsection{Restricted $\mathcal C$-Hull:  Number of Edges}\label{sec:restrictNumEdges}
Next, we show how to find a $\mathcal C$-hull $Q$ that balances its enclosed area and perimeter under the restriction that it consists of at most $k$ edges.  We say that $Q$ is optimal \emph{restricted to at most $k$ edges}, if there is no other $\mathcal C$-hull~$Q'$ with at most $k$ edges and $\cost(Q')<\cost(Q)$.
\vspace{0.3em}
\begin{ProblemEdgeShortcutHull}\label{problem2}\ \vspace{0.5ex} \\%
     \begin{tabular}{p{6ex}l}
        \textbf{given:}  & weakly-simple polygon $P$ with $n$ vertices\\& and connected exterior, set $\mathcal C$ of shortcuts\\& of $P$, $\lambda\in [0,1]$, and  $k\in \mathbb{N}$   \\
        \textbf{find:} &  optimal $\mathcal C$-hull $Q$ of $P$ (if it exists)\\ & restricted to at most $k$ edges. 
     \end{tabular}
\end{ProblemEdgeShortcutHull}

\noindent To solve \probEdgeHull we adapt Step 5 of the algorithm presented in Section~\ref{sec:basic-algorithm}. We extend the tables $\mathrm A$ and $\mathrm I$ by an additional dimension of size $k$ modelling the budget of edges that we have left for the particular instance. For a shortcut $e \in \mathcal C^+$ and a \emph{budget} $b$ we interpret the table entries as follows. 

\begin{compactitem}
        \item $\mathrm A[e][b]$: cost of labeled $\mathcal C^+$-triangulation~$\mathbf T$ of the pocket of $e$ s.t.\  $\mathbf T$ is optimal, the triangle adjacent to $e$ is active and $\mathbf T$ contains at most $b$ separators.
        \item $\mathrm{I}[e][b]$: cost of labeled $\mathcal C^+$-triangulation~$\mathbf T$ of the pocket of $e$ s.t.\ $\mathbf T$ is optimal, the triangle adjacent to $e$ is inactive and $\mathbf T$ contains at most $b$ separators.
\end{compactitem}

\noindent Let $e$ be the currently considered edge of $G_{\mathcal C^+}$ when filling the tables. 
For a triangle $\Delta=(e,e_1,e_2) \in \mathcal T_e$ of $e$ its active and inactive costs depend on the given budgets $b_1$ and $b_2$ with $1\leq b_1,b_2 \leq k$ that we intend to use for the sub-instances attached to $e_1$ and~$e_2$.  
\begin{align*}
    x_{\Delta,b_1,b_2} = \sum_{i\in\{1,2\}} \min\{\mathrm{A}[e_i][b_i],I[e_i][b_i-1]+\lambda\cdot \length(e_i)\}\\
    y_{\Delta,b_1,b_2}=\sum_{i\in\{1,2\}}\min\{\mathrm{A}[e_i][b_i-1]+\lambda\cdot \length(e_i),I[e_i][b_i]\}
\end{align*}
Hence, for the case that $e\in \mathcal C$ and $\mathcal T_e\neq \emptyset$ we define 
\[
   \mathrm A[e][b]=\min\{x_{\Delta,b_1,b_2} \mid \Delta \in \mathcal T_e,\,b_1+b_2=b\}+\beta\cdot\area(e),
\]
where $\beta=(1-\lambda)$. There are $b$ possible choices of $b_1$ and $b_2$ that satisfy $b_1+b_2=b$. Thus, we can compute $A[e][b]$ in $O(b)$ time. 
For the remaining cases we define
\[
   \mathrm A[e][b]=\begin{cases}
     \infty & \text{$e \not\in \mathcal C$} \\
     \beta\cdot \area(e) & \text{$e\in \mathcal C$, $\mathcal T_e=\emptyset$},\\
   \end{cases}
\] which can be computed in $O(1)$ time. 
Moreover, for the case that $e\not\in \mathcal C$ or $\mathcal T_e\neq \emptyset$ we define 
\[
  \mathrm I[e][b]=\min\{y_{\Delta,b_1,b_2} \mid \Delta \in \mathcal T_e,b_1+b_2=b\}.
\] For the same reasons as before we can compute $\mathrm I[e][b]$ in $O(1)$ time.
For $e \in \mathcal C$ or $\mathcal T_e\neq \emptyset$ we define $\mathrm I[e][b] = \infty$.  
Finally, to cover border cases we set $\mathrm A[e][0]=\infty$ and $\mathrm I[e][0]=\infty$.
Altogether, the entry $\mathrm I[e^\star][k]$ contains the cost of an optimal $\mathcal C$-hull that is restricted to $k$ edges. 
Apart from minor changes in Step~6 the other parts of the algorithm remain unchanged. 

\paragraph{Running time.}
Compared to the algorithm of Section~\ref{sec:basic-algorithm} the running time of computing a single entry increases by a factor of $O(k)$. Further, there are $O(k)$ times more entries to be computed, which yields that the running time increases by a factor of $O(k^2)$. 

\begin{thm}\label{thm:restricted-edges}
 The problem \probEdgeHull can be solved in $O(k^2 (\chi^3+ n \chi))$ time. In particular, it can be solved in $O(k^2 n^3)$ time in general and in $O(k^2 n)$ time if the edges in $\mathcal C$ do not cross.
\end{thm}

\subsection{Restricted $\mathcal C$-Hull: Number of Bends}
A slightly stronger constraint than restricting the number of edges is restricting the number of bends of a $\mathcal C$-hull. Formally, we call two consecutive edges of a simply-weakly polygon a \emph{bend} if the enclosed angle is not $180^\circ$.
We say that $Q$ is \emph{optimal restricted to at most $k$ bends} if there is no other $\mathcal C$-hull~$Q'$ with at most $k$ bends and $\cost(Q')<\cost(Q)$.
\vspace{0.3em}
\begin{ProblemBendShortcutHull}\label{problem3}\ \vspace{0.5ex} \\%
     \begin{tabular}{p{6ex}l}
        \textbf{given:}  & weakly-simple polygon $P$ with $n$ vertices\\& and connected exterior, set $\mathcal C$ of shortcuts\\& of $P$, $\lambda\in [0,1]$, and  $k\in \mathbb{N}$   \\
        \textbf{find:} &  optimal $\mathcal C$-hull $Q$ of $P$  (if it exists)\\ & that is restricted to at most $k$ bends. 
     \end{tabular}
\end{ProblemBendShortcutHull}

If the vertices of $P$ are in general position, i.e., no three vertices lie on a common line, a $\mathcal C$-hull $Q$ of $P$ is optimal restricted to at most $k$ bends if and only if it is optimal restricted to $k$ edges. Hence, in that case we can solve \probBendHull using the algorithm presented in Section~\ref{sec:restrictNumEdges}. 
In applications, the case that the vertices of $P$ are not in general position, occurs likely when the input polygon is, e.g., a schematic polygon or a polygon whose vertices lie on a grid.
In that case, we add an
edge $p_1p_h$ to $\mathcal C$ for each sequence $p_1,\dots,p_h$ of at least three vertices of $P$ that lie on a common line; we add $p_1p_h$ only if it lies in the exterior of $P$.
The newly obtained set $\mathcal C'$ has $O(n^2)$ edges. Hence, compared to $\mathcal C$ it possibly has an increased spatial complexity with $\chi\in O(n)$. From Theorem~\ref{thm:restricted-edges} we obtain the next result.

\begin{thm}
 The problem \probBendHull can be solved in $O(k^2 \cdot  n^3)$ time. 
\end{thm}

\section{Relations to other Geometric Problems}\label{sec:applications}
 We have implemented the algorithm presented in Section~\ref{sec:basic-algorithm}. For example, computing a shortcut hull for the instance shown in Figure~\ref{fig:example:lakes} one run of the dynamic programming approach (Step 5) took $400$ms on average. 
 This suggests that despite its cubic worst-case running time our 
algorithm is efficient enough for real-world applications. However, more 
experiments are needed to substantiate this finding.

\paragraph{Balancing the Costs of Area and Perimeter}
\begin{figure}[!ht]
    \centering
  \begin{subfigure}[t]{0.49\linewidth}
    \includegraphics[width=\textwidth]{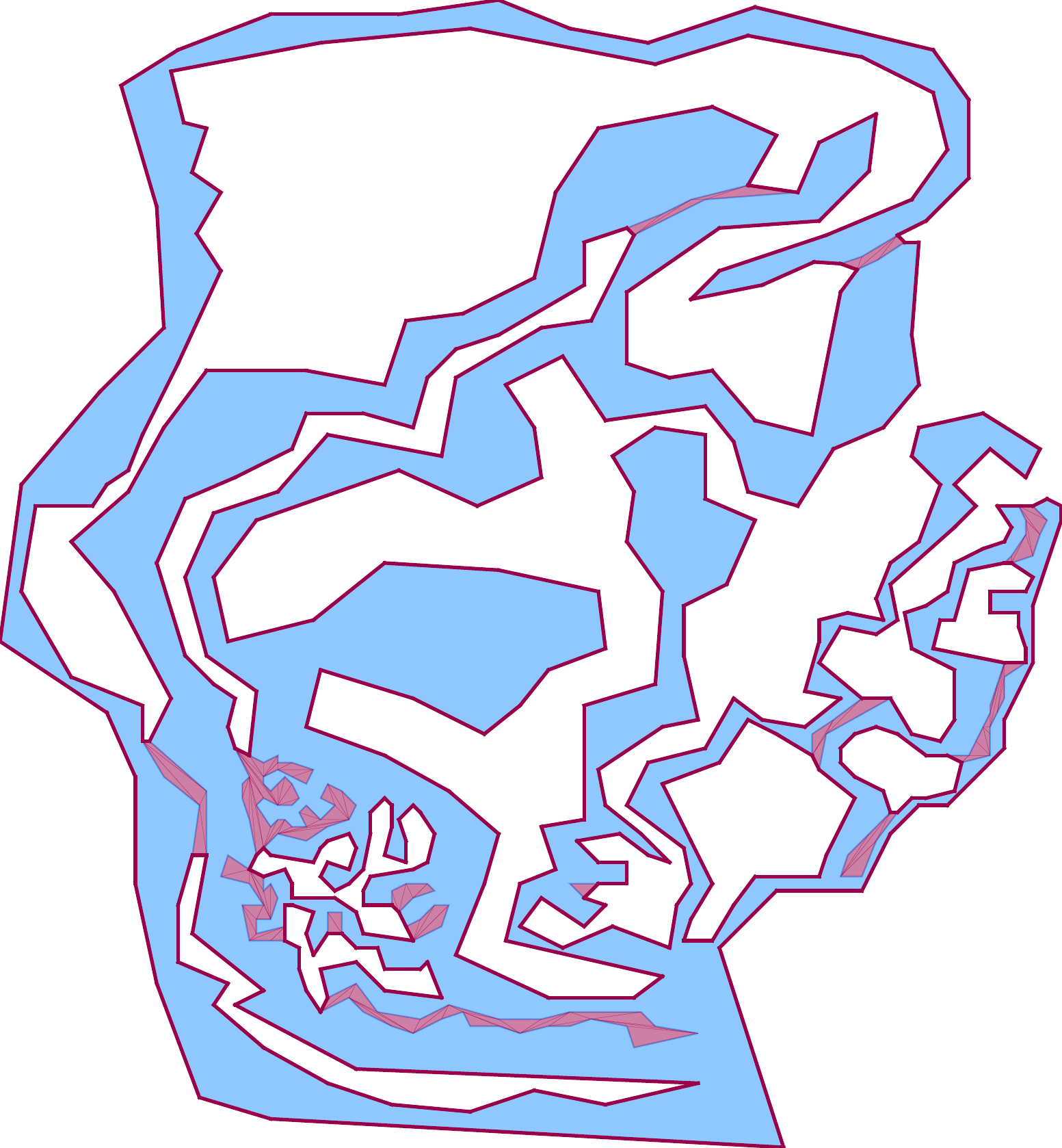}
    \caption{}
    \label{fig:balancingAlphaA}
  \end{subfigure}
    \begin{subfigure}[t]{0.49\linewidth}
    \includegraphics[width=\textwidth]{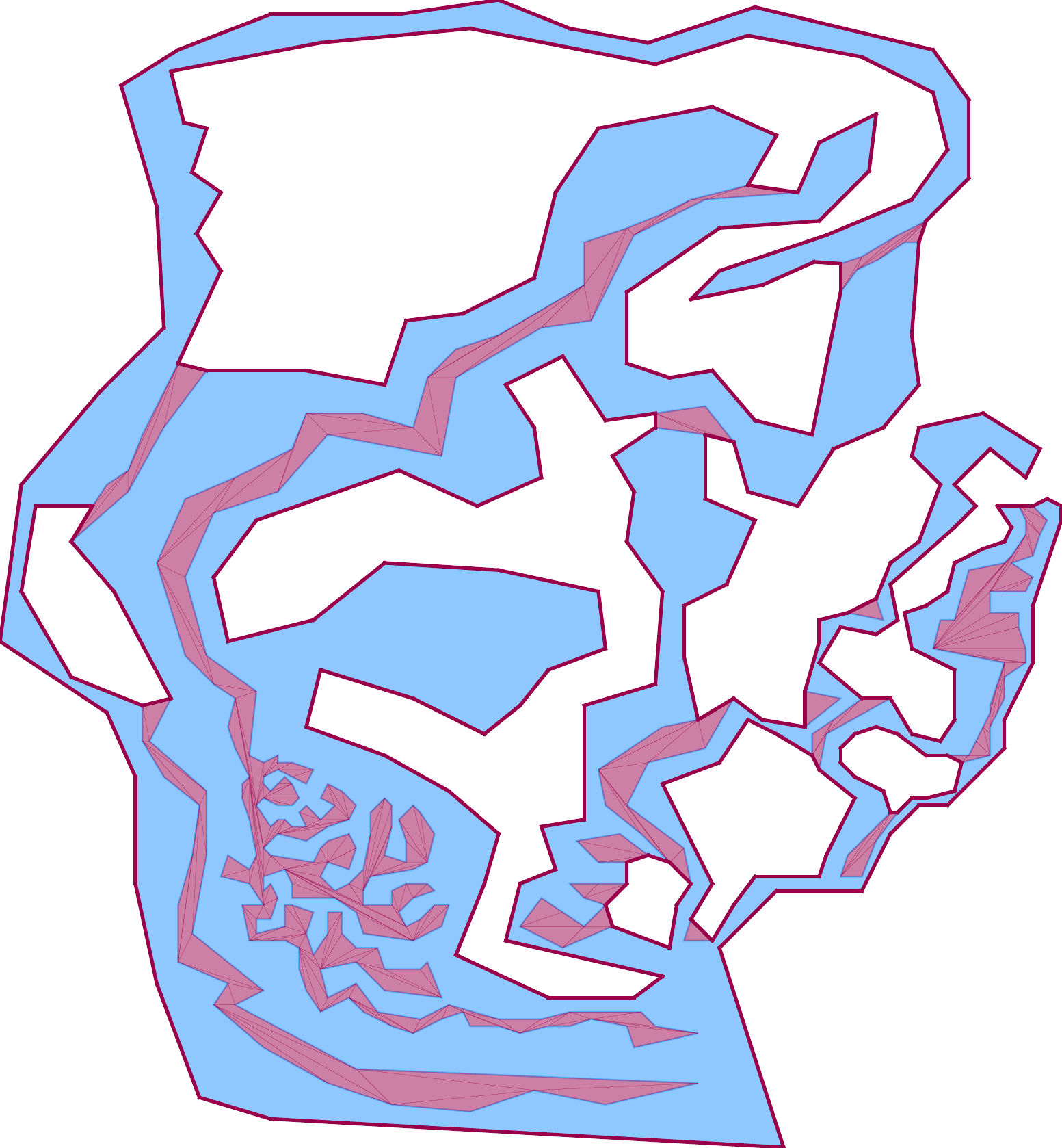}
    \caption{}
    \label{fig:balancingAlphaB}
    \end{subfigure}
    
    \begin{subfigure}[t]{0.49\linewidth}
    \includegraphics[width=\textwidth]{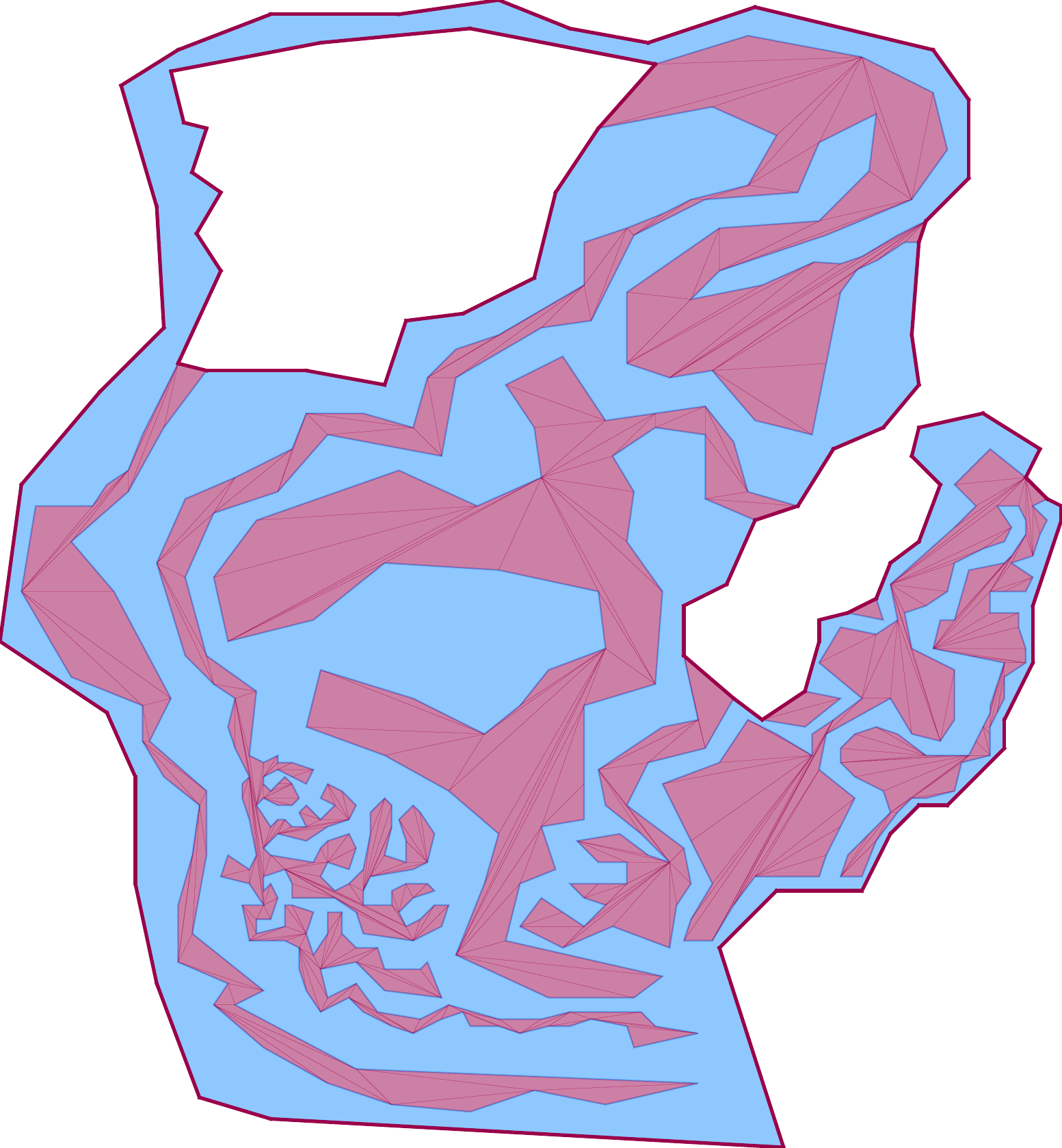}
    \caption{}
    \label{fig:balancingAlphaC}
    \end{subfigure}
    \begin{subfigure}[t]{0.49\linewidth}
    \includegraphics[width=\textwidth]{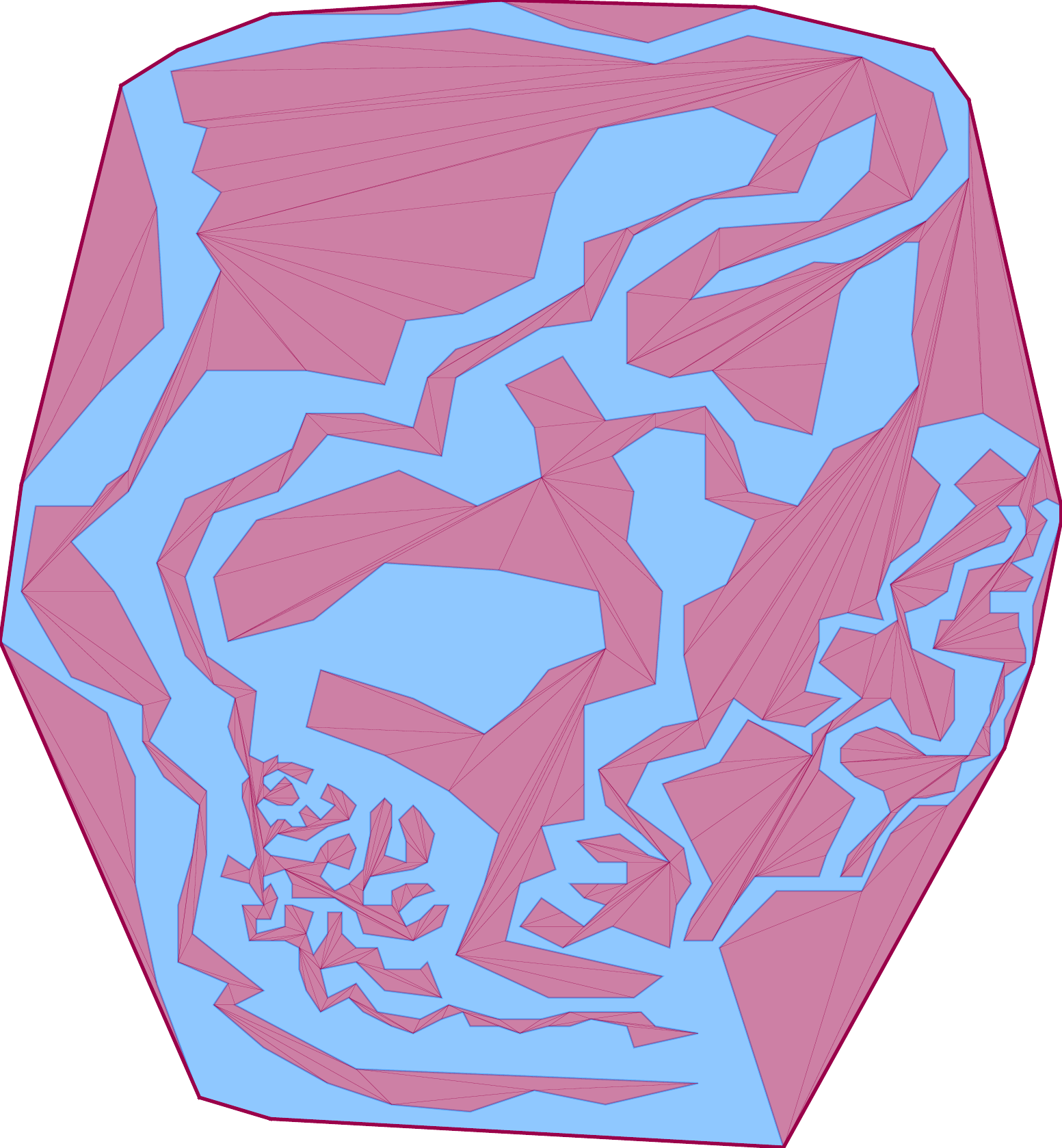}
    \caption{}
    \label{fig:balancingAlphaD}
    \end{subfigure}
    \caption{Optimal $\mathcal C$-hulls for increasing values of $\lambda$. In particular, for $\lambda = 1$ we only consider the costs for the area (Subfigure (a)) and for $\lambda = 0$ we only consider the costs for the perimeter (Subfigure (d)). }
    \label{fig:balancingAlpha}
\end{figure}
In Figure~\ref{fig:example:output}
we display a series of optimal 
$\mathcal C$-hulls\footnote{Figure~\ref{fig:example:output}b: $\lambda=0.906$; 
Figure~\ref{fig:example:output}c: $\lambda=0.995$; 
Figure~\ref{fig:example:output}e: $\lambda=0.914$; 
Figure~\ref{fig:example:output}f: $\lambda=0.975$}; see also Figure~\ref{fig:balancingAlpha}.
We use the same polygon and the set of all 
possible shortcuts as input while increasing the parameter 
$\lambda$ of the cost function. 
To find relevant values of $\lambda$
we implemented a systematic search in the range $[0,1]$.  It uses the simple observation that with monotonically increasing $\lambda$ the amount of area enclosed by an optimal shortcut hull increases monotonically. 
More in detail, we compute the 
optimal shortcut hull for $\lambda=0$ and $\lambda=1$. 
If the area cost $c_A$ of these shortcut hulls differ, we recursively consider  the intervals $[0,0.5]$ and $[0.5,1]$ for the choice of~$\lambda$ similar to a binary search. 
Otherwise, we stop the search.

As presented in \autoref{equ:costs}, we consider costs 
for the area and perimeter in \textsc{ShortcutHull}. 
The second column of Figure~\ref{fig:example:output} shows a result for a small value of $\lambda$, i.e., the costs for the area are weighted higher. 
As expected the resulting optimal $\mathcal C$-hull 
is rather close to the input polygon.
In contrast, the last column of Figure~\ref{fig:example:output} shows the optimal $\mathcal C$-hull for a larger $\lambda$-value. 
We particularly obtain holes that represent large areas enclosed by the polygon, while small gaps are filled.

\paragraph*{Simplification and Schematization of Simple Polygons}
\begin{figure}[t]
    \centering
  \begin{subfigure}[t]{0.30\linewidth}
    \includegraphics[width=\textwidth]{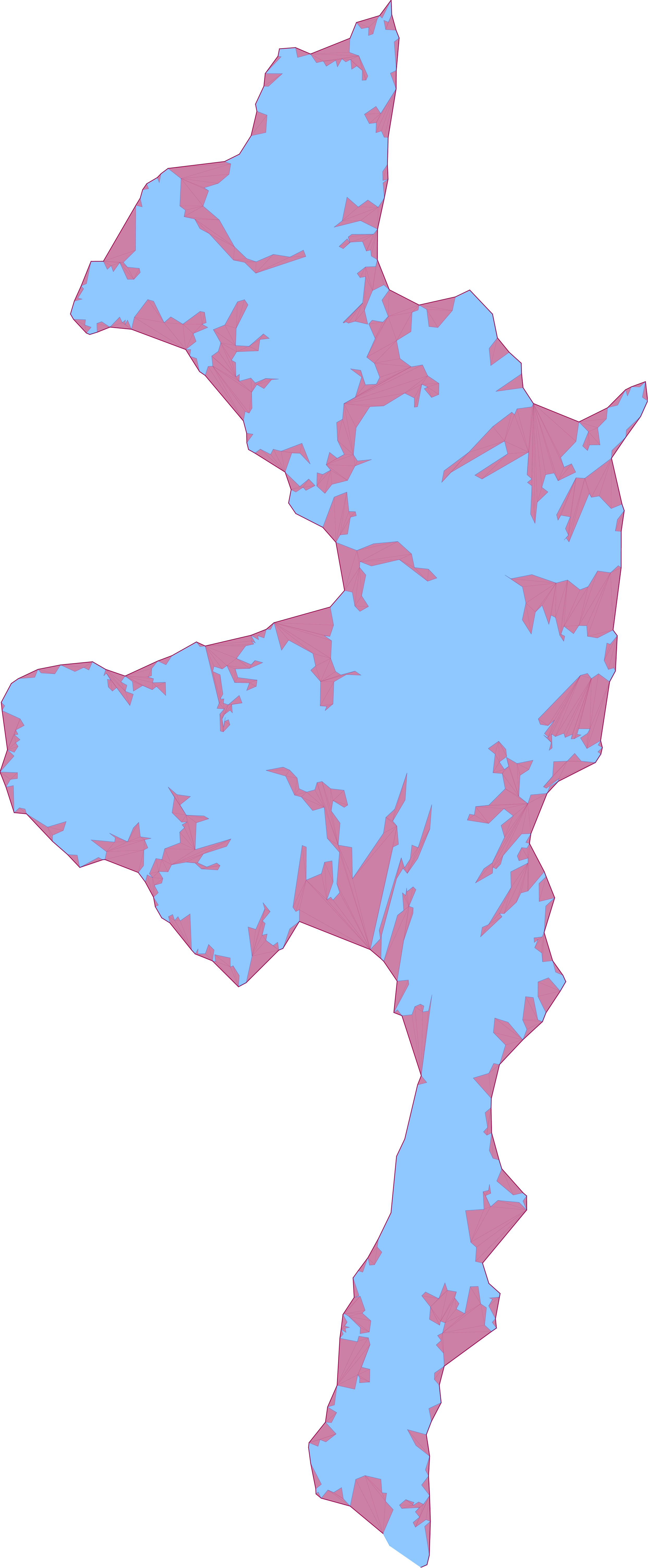}
    \caption{}
    \label{fig:shetlandA}
  \end{subfigure}
    \begin{subfigure}[t]{0.30\linewidth}
    \includegraphics[width=\textwidth]{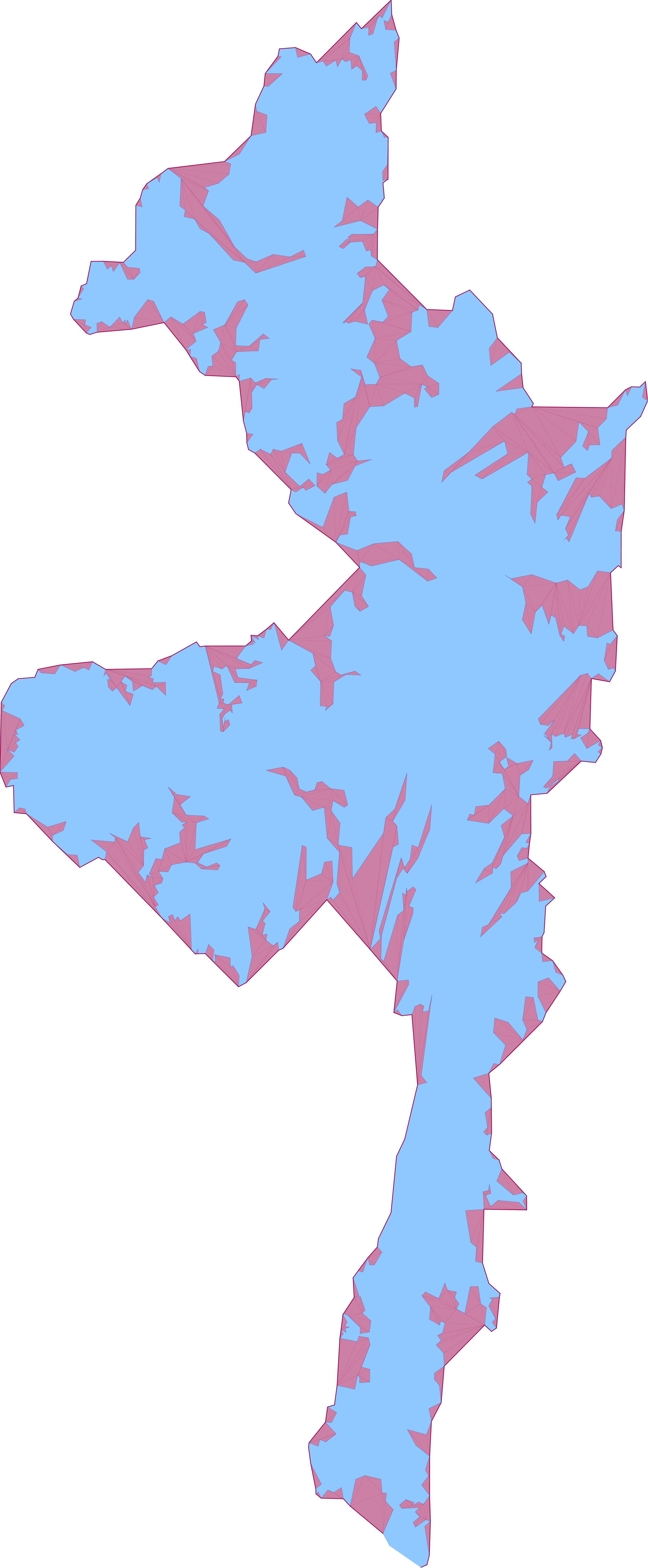}
    \caption{octilinear}
    \label{fig:shetlandB}
    \end{subfigure}
    \begin{subfigure}[t]{0.30\linewidth}
    \includegraphics[width=\textwidth]{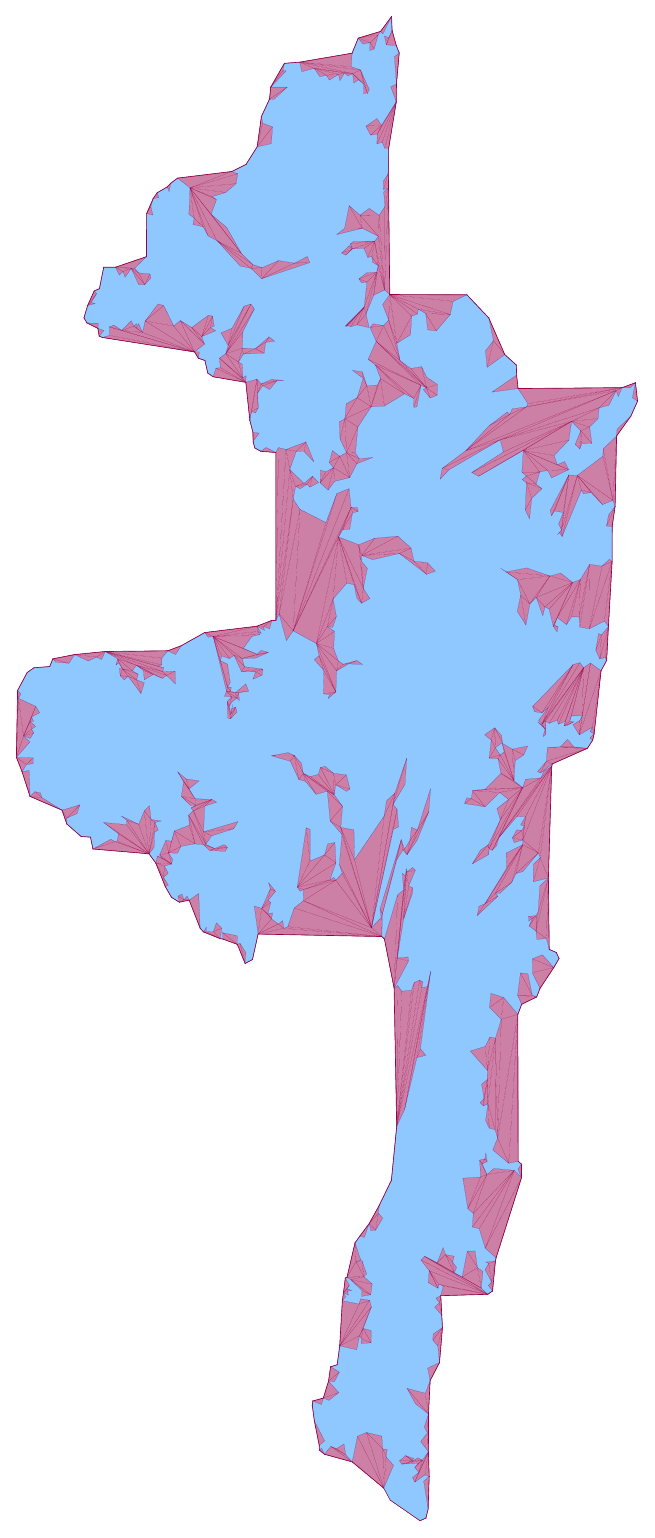}
    \caption{rectilinear}
    \label{fig:shetlandC}
    \end{subfigure}
    \caption{Simplification (a) and schematization (b)--(c) of the main island of Shetland.}
    \label{fig:shetland}
\end{figure}
In the following, we discuss how our approach relates to typical measures for simplification and schematization. These are the number of edges, the number of bends~\cite{douglas1973algorithms} or the perimeter~\cite{tufte1985visual}, which are implemented by shortcut hulls; e.g., Figure~\ref{fig:shetlandA} shows the simplification of the border of the main island of Shetland by a $\mathcal C$-hull as defined in \textsc{ShortcutHull}. 
The schematization of a polygon is frequently implemented as a hard constraint with respect to a given set~$O$ of edge orientations. 
For schematizing a polygon with $\mathcal C$-hulls, we outline two possibilities: a non-strict and a strict schematization. 
For the non-strict schematization,  we adapt the cost function of the shortcuts such that edges with an orientation similar to an orientation of $O$ are cheap while the others are expensive; see Figure~\ref{fig:shetlandB} for $O$ consisting of horizontal, vertical, and diagonal orientations and Figure~\ref{fig:shetlandC} for $O$ consisting of the horizontal and vertical orientations. 
The strict schematization restricts the set $\mathcal C$ of shortcuts, such that each edges' orientation is from $O$. 
For example, one can define $\mathcal C$ based on an underlying grid that only uses orientations from $O$. 
We then 
need to take special care about the connectivity of $\mathcal C$, e.g., by also having all edges of the input polygon in~$\mathcal C$. 

\begin{figure}[t]
    \centering
  \begin{subfigure}[t]{0.45\linewidth}
    \includegraphics[height=\textwidth,angle=90]{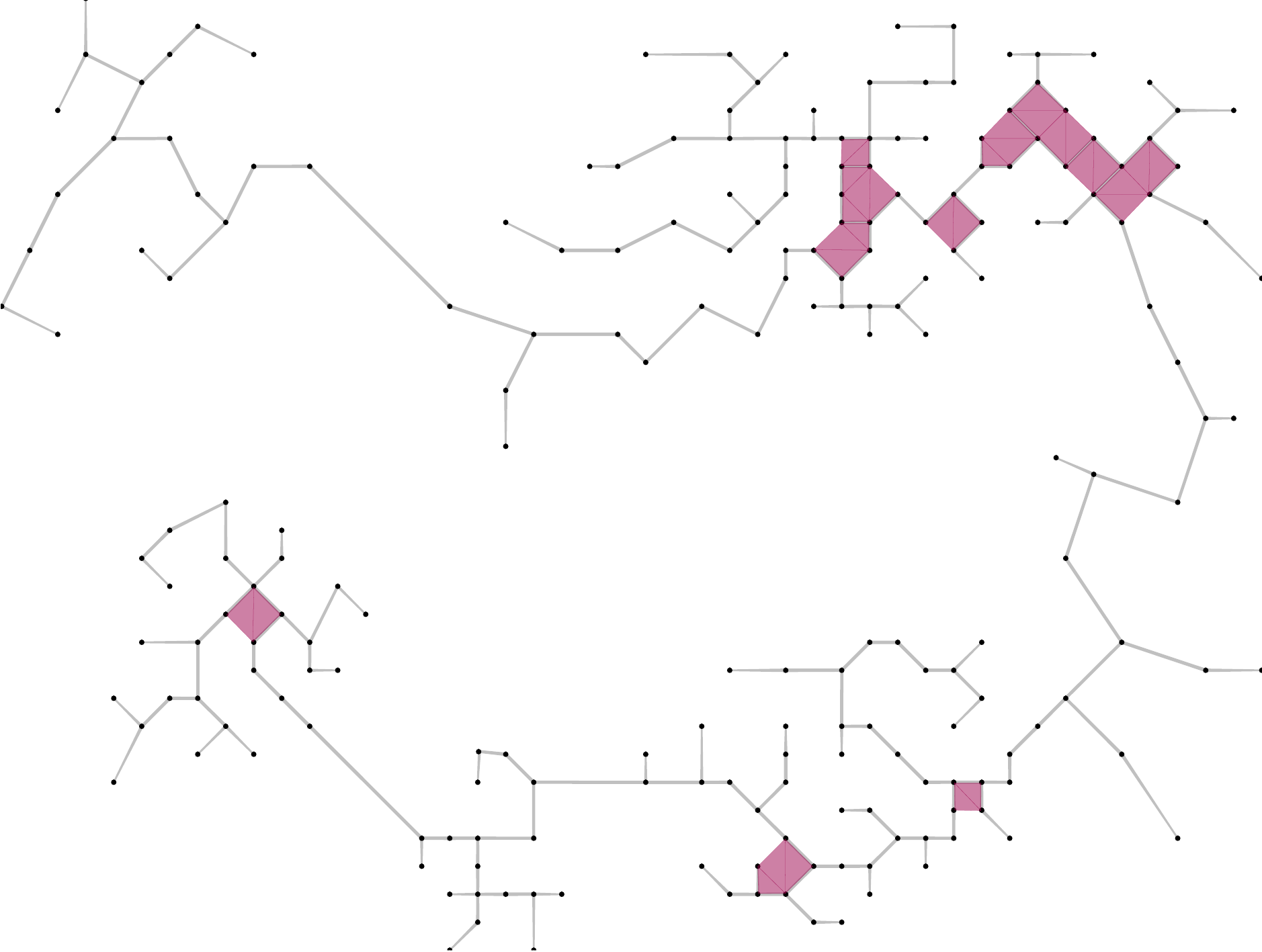}
    \caption{}
    \label{fig:pointsetA}
  \end{subfigure}
    \hspace{1em}
    \begin{subfigure}[t]{0.45\linewidth}
    \includegraphics[height=\textwidth,angle=90]{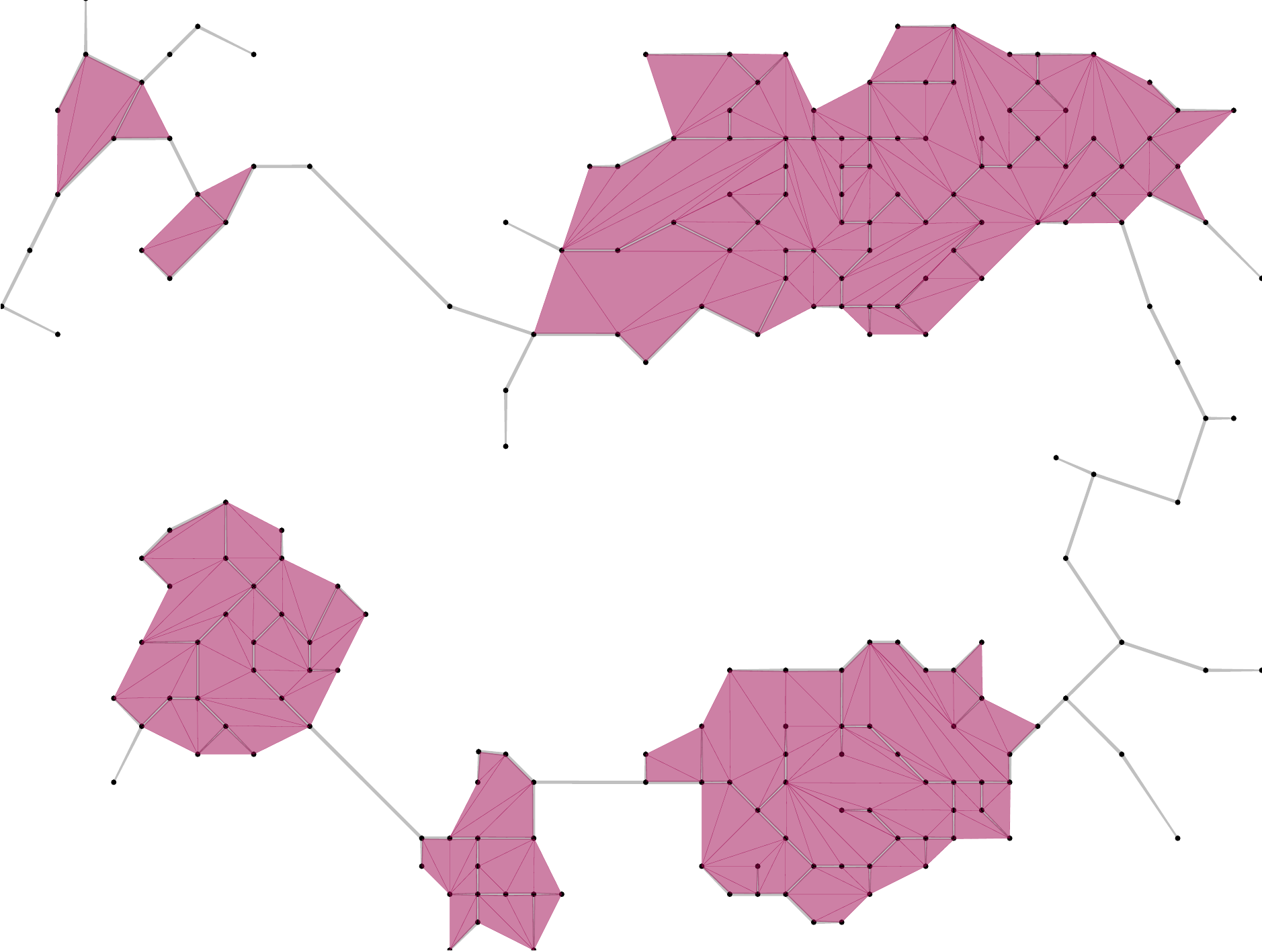}
    \caption{}
    \label{fig:pointsetD}
    \end{subfigure}
    
    \begin{subfigure}[t]{0.45\linewidth}
    \includegraphics[height=\textwidth,angle=90]{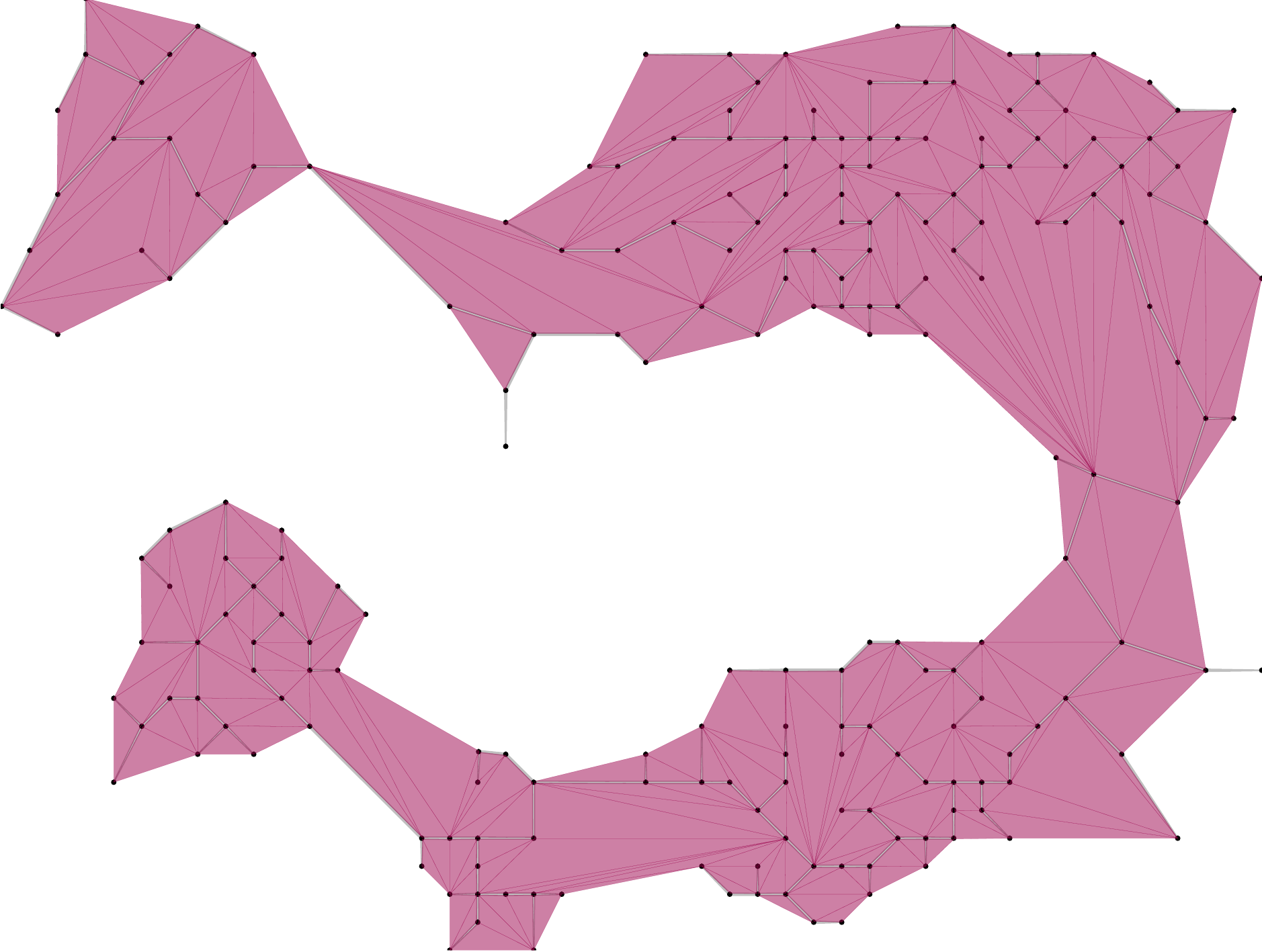}
    \caption{}
    \label{fig:pointsetE}
    \end{subfigure}
    \hspace{1em}
    \begin{subfigure}[t]{0.45\linewidth}
    \includegraphics[height=\textwidth,angle=90]{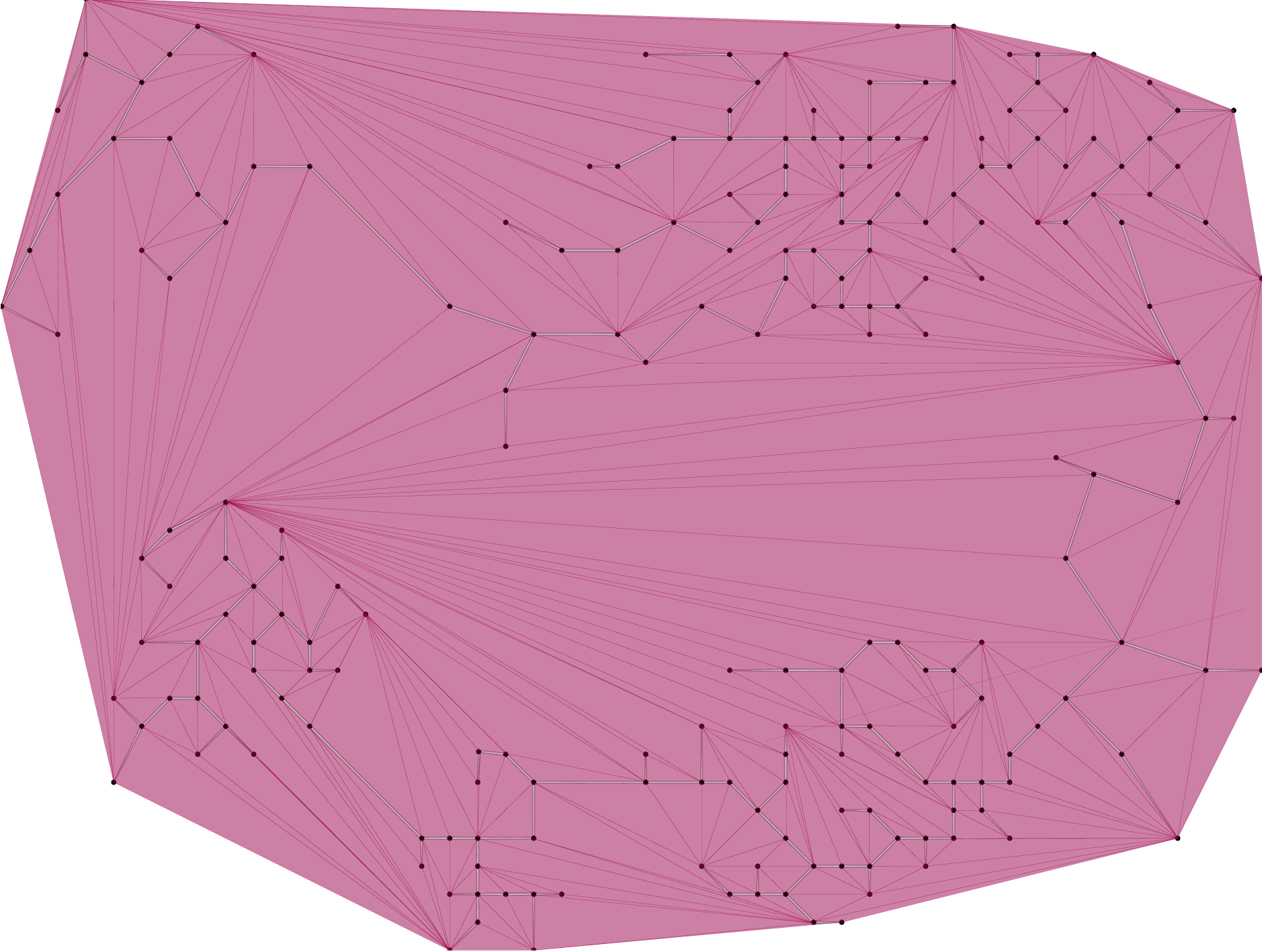}
    \caption{}
    \label{fig:pointsetF}
    \end{subfigure}
    \caption{Optimal $\mathcal C$-hulls for increasing values of $\lambda$ for a point set using a minimum spanning tree as basis.}
    \label{fig:pointset}
\end{figure}

\paragraph{Aggregation of Multiple Objects and Clustering}
We can adapt $\mathcal C$-hulls 
for multiple geometric objects, e.g. a point set. 
We suggest to use a geometric graph that contains all vertices of the input geometries, all edges of the input geometries and is connected as input for problem \textsc{ShortcutHull}, e.g., a minimum spanning tree of the point set; see~Fig~\ref{fig:pointset}. 
With increasing $\lambda$-value the regions of the shortcut hull first enclosed are areas with high density.  
By removing all edges of $Q$ that are not adjacent to the interior of $Q$, we possibly receive multiple polygons which each can be interpreted as a cluster.

\section{Conclusion}
We introduced a simplification technique for polygons that yields shortcut hulls, i.e., crossing-free polygons that are described by shortcuts and that enclose the input polygon. 
In contrast to other work, we consider the shortcuts as input.
We introduced a cost function of a shortcut hull that is a linear combination of the covered area and the perimeter. 
Computing optimal shortcut hulls without holes takes $O(n^2)$ time.
For the case that we permit holes we presented an algorithm based on dynamic programming that runs in $O(n^3)$ time. If the input shortcuts do not cross it runs in $O(n)$ time. 

We plan on considering (i) the bends as part of the cost function, (ii) more general shortcuts, e.g. allowing one bend per shortcut, and (iii) optimal spanning trees for the case of multiple input geometries.

\paragraph{Acknowledgements}
This work has partially been funded by the German Research Foundation under Germany's Excellence Strategy, EXC-2070\,-\,390732324\,-\,PhenoRob, and
by NSF (Mitchell, CCF-2007275).

\small
\bibliographystyle{abbrv}

\bibliography{paper}

\end{document}